\documentclass{uai2021} % for initial submission
\usepackage[american]{babel}
\usepackage{natbib} % has a nice set of citation styles and commands
\usepackage{mathtools} % amsmath with fixes and additions
\usepackage{booktabs} % commands to create good-looking tables
\usepackage{tikz} % nice language for creating drawings and diagrams

\usepackage{graphicx}
\usepackage{caption}
\usepackage{algorithm}
\usepackage{algorithmic}
\usepackage{amsmath}
\usepackage{amssymb}
\usepackage{amsthm}
\usepackage{booktabs}       % professional-quality tables
\usepackage{amsfonts}       % blackboard math symbols
\usepackage{nicefrac}       % compact symbols for 1/2, etc.
\usepackage{microtype}      % microtypography
\usepackage{subfig}
\usepackage{bm}

\usepackage{diagbox}
\usepackage{multicol, multirow}

\makeatletter
\newcommand{\rmnum}[1]{\romannumeral #1}
\newcommand{\Rmnum}[1]{\expandafter\@slowromancap\romannumeral #1@}
\makeatother
\newtheorem{theorem}{\textbf{Theorem}}
\newtheorem{lemma}{\textbf{Lemma}}

\newtheorem{definition}{\textbf{Definition}}

\newtheorem{remark}{\textbf{Remark}}
\newtheorem{proposition}{Proposition}

\makeatletter
\def\saveenum{\xdef\@savedenum{\the\c@enumi\relax}}
\def\resetenum{\global\c@enumi\@savedenum}
\makeatother

\title{Random Projection and Recovery for High Dimensional Optimization With Arbitrary Outliers}
\author[]{Hu Ding} % Lead author
\author[]{Jiawei Huang}
\author[]{Ruizhe Qin}
\author[]{Fan Yang}
\affil[]{%
	School of Computer Science and Technology\\
	University of Science and Technology of China\\
	He Fei, China \\
}

\begin{document}

\maketitle

\begin{abstract}
%In this paper, we consider robust optimization problems in high dimensions. Because a real-world dataset may contain significant noise or even specially crafted samples from some attacker, we are particularly interested in the optimization problems with arbitrary (and potentially adversarial) outliers.
%We focus on two fundamental optimization problems: {\em SVM with outliers} and {\em $k$-center clustering with outliers}. They are in fact extremely challenging combinatorial optimization problems, since we cannot impose any restriction on the adversarial outliers. Therefore, their computational complexities are quite high especially when we consider the instances in high dimensional spaces.
%A real-world dataset often contains significant noise or even specially crafted samples from some attacker, and thus the optimization with arbitrary  outliers 
Robust optimization problems have attracted considerable attention in recent years.  In this paper, we focus on two fundamental robust optimization problems: {\em SVM with outliers} and {\em $k$-center clustering with outliers}. The key obstacle is that the outliers can be located arbitrarily in the space ({\em e.g.,} by an attacker), and thus they are actually quite challenging combinatorial optimization problems. Their computational complexities can be very high especially in high dimensional spaces. 
 The {\em Johnson-Lindenstrauss (JL) Transform} is a popular random projection  method for dimension reduction. Though the JL transform has been widely studied in the past decades, its effectiveness for dealing with high-dimensional optimizations with outliers has never been investigated before (to the best of our knowledge).
Based on some novel insights from the geometry, we prove that the complexities of these two problems can be significantly reduced through the JL transform. Moreover, we prove that the solution in the dimensionality-reduced space can be efficiently recovered in the original $\mathbb{R}^d$ while the quality is still preserved. To study its performance in practice, we compare different JL transform methods with  several other well known dimension reduction methods in our experiments.
\end{abstract}

\vspace{-0.1in}

\section{Introduction}
\label{sec-intro}
\vspace{-0.05in}

{\em Johnson-Lindenstrauss (JL) Transform} is a popular random projection method for solving high-dimensional problems  in different areas~\citep{jltutorial}. 
%, such as optimization~\cite{sarlos2006improved}, compressive sensing~\cite{baraniuk2006johnson}, and privacy preserving~\cite{DBLP:conf/focs/BlockiBDS12}.
Compared with the {\em data-aware} dimension reduction techniques ({\em e.g.,} PCA and feature selection)~\citep{cunningham2015linear}, JL transform is a {\em data-oblivious} technique that is more convenient to implement in practice, especially for distributed computing and streaming data. 

\begin{lemma}[\citep{johnson1984extensions}]
	\label{lem-jl}
	Given $\epsilon>0$, $n\in \mathbb{Z}^+$, and $\tilde{d}=O(\frac{\log n}{\epsilon^2})$, for any set $P$ of $n$ points in $\mathbb{R}^d$, there exists a mapping $f$: $\mathbb{R}^d\rightarrow \mathbb{R}^{\tilde{d}}$ such that for all $p, q\in P$, 
	\begin{eqnarray}
		\Big|||p-q||^2- ||f(p)-f(q)||^2\Big|\leq \epsilon||p-q||^2.\label{for-jl}
	\end{eqnarray}
\end{lemma}
Here, we use $||\cdot||$ to denote the Euclidean distance. \textbf{The mapping $f$ in Lemma~\ref{lem-jl} is called the ``JL transform''. }

The lemma shows that only $O(\frac{\log n}{\epsilon^2})$ dimensions, which is independent of $d$, are sufficient to approximately preserve their pairwise distances. 
In the past decades, a large amount of articles focused on the construction of ``$f$''. Actually, a simple way is to build a $\tilde{d}\times d$ matrix $A$ where each entry is an independent Gaussian $\mathcal{N}(0,1)$ random variable~\citep{dasgupta2003elementary}. 
\cite{achlioptas2003database} proposed a much easier implementation with entries belonging in $\{0, \pm 1\}$.
%; so its implementation is much easier especially for using SQL in a database environment. 
Inspired by {\em Heisenberg’s Uncertainty Principle}, \cite{ailon2009fast} provided a faster implementation  via {\em Walsh-Hadamard} matrix. Furthermore, a number of improvements on the sparsity of the matrix $A$ have been proposed, such as \citep{DBLP:journals/jacm/KaneN14,DBLP:conf/stoc/DasguptaKS10}.

\vspace{-0.12in}
\subsection{Our Contributions}
\label{sec-our}
\vspace{-0.08in}
In this paper, we study the effectiveness of JL transform   for the optimization problems with outliers in high dimensions. 
%Our motivations for considering outliers are twofold. \textbf{(\rmnum{1})} 
A real-world  dataset often contains a significant amount of outliers that can seriously affect the final result in  machine learning~\citep{zimek2012survey}.
%\textbf{(\rmnum{2})} The field of {\em  adversarial machine learning} concerning about the potential vulnerabilities has attracted a great amount attentions in recent years~\cite{biggio2018wild}. For example, 
Moreover, an attacker can inject a small set of outliers to the dataset so as to make the decision boundary severely deviate and cause unexpected mis-classification~\citep{biggio2018wild}. Therefore, designing robust machine learning algorithms  are urgently needed to meet these challenges. We consider two fundamental optimization problems in high dimensions:  \textbf{support vector machine (SVM) with outliers} and \textbf{$k$-center clustering with outliers}. Also, we assume that the outliers can be located arbitrarily   in the space ({\em e.g.,} by an attacker). Both of the two problems have many important applications for classification and data analysis~\citep{tan2006introduction}.  

Because we cannot impose any restriction on the outliers, existing methods for these two problems often have high complexities (more related works are shown in Section~\ref{sec-relate}).  In the past years, we are aware that JL transform has been widely studied for reducing the time complexities for the clustering~\citep{kerber2014approximation,boutsidis2014randomized,cohen2015dimensionality,makarychev2019performance} and SVM~\citep{DBLP:journals/ml/ArriagaV06,balcan2006kernels,kumar2008randomized,DBLP:conf/icml/ShiSHH12,DBLP:journals/tkdd/PaulBMD14} problems. \textbf{But its effectiveness for the case with arbitrary outliers has never been studied, to the best of our knowledge.}

Another key issue concerned in this paper is the ``recovery'' step. One may seek to recover   the solution from the dimensionality-reduced space to the original $\mathbb{R}^d$ with preserving the quality ({\em e.g.,} the radius for $k$-center clustering or the separating margin for SVM). The recovery step is important for clustering/classifying new coming data and extracting other useful information from the original space~\citep{zhang2013recovering}. For example, the normal vector of the SVM separating hyperplane in the original input space can provide us useful information on the features and their relations to classification. It is worth to emphasize that most of the previous articles on JL transform (even only for the case without outliers, {\em e.g.,} \citep{balcan2006kernels,DBLP:conf/icml/ShiSHH12,DBLP:journals/tkdd/PaulBMD14}) are in lack of the discussion on the recovery step. 

To analyze the quality for dimension reduction and recovery, \textbf{the major challenge} is that the outliers and inliers are mixed and an outlier can even become an inlier (and vice versa) in the dimensionality-reduced space (this is also the key difference that has never been considered in the previous articles on dimension reduction). 
In this paper, 
%we focus on resolving these issues. Specifically, 
we provide a unified framework for dimension reduction and recovery with provable quality guarantee. \textbf{The framework can be summarized as the following three steps:} 
\vspace{-0.1in}
 \begin{enumerate}
\item Apply the JL transform $f$ to reduce the dimensionality of the input dataset $P$.
\item Run any existing algorithm $\mathcal{A}$ in the dimensionality-reduced space and obtain the solution $\bf{\bar{x}}$.
\item Recover the solution $f^{-1}(\bf{\bar{x}})$ in the original space by using the sparse approximation techniques. 
\end{enumerate}
\vspace{-0.1in}

%Under such a framework, the complexities can be significantly reduced; furthermore, their solutions in the dimensionality-reduced space can be efficiently recovered in the original space. 

%\noindent\textbf{The unified framework for dimension reduction and recovery:}

%\begin{algorithm}[tb]
%   \caption{The Framework for Dimension Reduction and Recovery}
%   \label{alg-framework}
%\begin{algorithmic}
% \STATE {\bfseries Input:} An instance $P$ of SVM with outliers or $k$-center clustering with outliers.
% %   \STATE {\bfseries Input:} A point set $S$ in $\mathbb{R}^d$, the origin $o$, and $N\in \mathbb{Z}^+$.
%%    \begin{enumerate}
%%   \item Initialize $i=1$ and $v_1$ to be the closest point in $S$ to $o$.
%%   \item Iteratively perform the following steps until $i=N$.
%%   \begin{enumerate}
%%   \item Find the point $p_i \in S$ who has the smallest projection distance to $o$ on $\overline{ov_{i}}$, {\em i.e.,} $p_i=\arg\min_{p\in S}\{||p\mid_{v_i}||\}$
%%   (see Figure~\ref{fig-gilbert}).
%%   \item Let $v_{i+1}$ be the point on segment $\overline{v_i p_i}$ closest to the origin $o$; update $i=i+1$.
%%    \end{enumerate}
%%    \end{enumerate}
%%        \STATE {\bfseries Output:} $v_N$ as an approximate solution of the polytope distance between $o$ and $S$.    
%\end{algorithmic}
%\end{algorithm}

The whole complexity consists of three parts corresponding to the above three steps. In practice, the time complexities of the first and third steps are neglectable. Also, since we reduce the dimension in  step 2, the whole complexity is reduced consequently. We will discuss the time complexities in Section~\ref{sec-svm} and \ref{sec-meb} with more details. 
Another highlight is that any existing algorithm for SVM with outliers or $k$-center clustering with outliers can be used as the black box in step~2. Hence the efficiency of our framework can be always improved if any new  algorithm is proposed in future.

\vspace{-0.1in}
%\vspace{-0.05in}
\subsection{Other Related Works}
\label{sec-relate} 
\vspace{-0.05in}

\textbf{SVM and robust SVM.} SVM is a popular model for classification~\citep{journals/tist/ChangL11}. A number of methods  have been proposed, such as \citep{mach:Cortes+Vapnik:1995,platt99,conf/nips/CrispB99,bb57389}, for the SVM problem. One-class SVM has also been used as an efficient approach for anomaly detection~\citep{DBLP:journals/pr/ErfaniRKL16,DBLP:conf/nips/ScholkopfWSSP99}.  In recent years, a number of algorithms were also developed for dealing with the case that a significant fraction of outliers are mixed in the dataset~\citep{conf/aaai/XuCS06,yang2010relaxed,icml2014c2_suzumura14,ding2015random,DBLP:journals/pr/XuCHP17}.  But these algorithms usually are super-linear and cannot efficiently deal with large datasets in high dimensions. Also, several defending strategies for SVM against adversarial attacks were summarized in the recent survey paper~\citep{biggio2018wild}.

Beyond the aforementioned JL transform based dimension reduction methods in Section~\ref{sec-our}, a number of other methods ({\em e.g.,} feature selection) were also proposed for speeding up  SVM training~\citep{rahimi2008random,paul2015feature}. However, to the best of our knowledge, it is still unclear about their quality guarantees in the presence of outliers.

\textbf{$k$-center clustering with outliers.} \cite{charikar2001algorithms} proposed a $3$-approximation algorithm for $k$-center clustering with outliers in arbitrary metrics. 
The time complexity of their algorithm is quadratic in data size. A following streaming $(4+\epsilon)$-approximation algorithm was proposed by \citep{mccutchen2008streaming}. 
\cite{DBLP:conf/icalp/ChakrabartyGK16} gave a $2$-approximation algorithm for metric $k$-center clustering with outliers based on the LP relaxation techniques. 
\cite{BHI} proposed a coreset based approach but having an exponential time complexity if $k$ is not a constant. Recently, \cite{DBLP:conf/esa/DingYW19} provided a greedy algorithm that yields a bi-criteria approximation (returning more than $k$ clusters). The distributed algorithms include~\citep{malkomes2015fast,DBLP:journals/pvldb/CeccarelloPP19}. 

\vspace{-0.15in}

\section{Preliminaries}
\label{sec-pre}
\vspace{-0.1in}

%Below, we introduce several definitions that will be used in the paper.  
As mentioned before, we do not impose any restriction on the outliers. We define the following combinatorial optimization problems following the  ``\textbf{trimming}'' idea from robust statistics~\citep{books/wi/RousseeuwL87}. 

%\begin{definition}[\textbf{Minimum Enclosing Ball (MEB)}]
%\label{def-meb}
%Given a set $P$ of $n$ points in $\mathbb{R}^d$, the MEB problem is to find a ball with minimum radius to cover all the points in $P$. The resulting ball and its radius are denoted by $\mathtt{MEB}(P)$ and $\mathtt{rad}(P)$, respectively.
%\end{definition}

For the SVM problem, we consider one-class and two-class separately. Let $o$ be the origin of the $\mathbb{R}^d$ space. Given a point $x$, denote by $\mathcal{H}_x$ the hyperplane passing through $x$ and being orthogonal to the vector $x-o$. We define $\mathcal{H}_x^+$ to be the closed half space that is bounded by $\mathcal{H}_x$ and excludes the origin $o$. We use $|\cdot|$ to denote the size of a given point set.

\begin{definition}[\textbf{One-class SVM with Outliers}]
\label{def-oneclass}
Given a set $P$ of $n$ points in $\mathbb{R}^d$ and $\gamma\in (0,1)$, the problem of one-class SVM with outliers is to find a point $x$ such that $|P\cap \mathcal{H}_x^+|=(1-\gamma)n$ and the (Euclidean) norm $||x||$ is maximized. 
``$||x||$'' actually is the margin width  separating $o$ and $\mathcal{H}_x$.
\end{definition}

\begin{definition} [\textbf{Two-class SVM with Outliers}]
\label{def-twoclass}
 Given two point sets $P_1$ and $P_2$ in $\mathbb{R}^d$ and two small parameters $\gamma_1, \gamma_2\in (0,1)$, the problem of two-class SVM with outliers is to find two subsets $P'_1\subset P_1$ and $P'_2\subset P_2$ with $|P'_1|=(1-\gamma_1)|P_1|$ and $|P'_2|=(1-\gamma_2)|P_2|$, and a margin separating  $P'_1$ and $P'_2$, such that the width of the margin is maximized. The margin is the region bounded by two parallel hyperplanes.
%  $\mathcal{H}^\perp$ and $\mathcal{H}^\top$.
\end{definition}
%\vspace{-0.1in}
 
%Obviously, to determine the margin for one-class (two-class) SVM, we just need to find the normal vector $x$. 

\vspace{-0.05in}

\begin{definition}[\textbf{$k$-Center Clustering with Outliers}]
\label{def-outlier}
Given a set $P$ of $n$ points in $\mathbb{R}^d$, $k\in\mathbb{Z}^+$, and $\gamma\in (0,1)$, $k$-center clustering with outliers is to find a subset $P'\subseteq P$, where $|P'|=(1-\gamma)n$, and $k$ centers $\{c_1, \cdots, c_k\}\subset \mathbb{R}^d$, such that $\max_{p\in P'}\min_{1\leq j\leq k}||p-c_j||$ is minimized. 
\end{definition}
We use $\mathbb{B}(c, r)$ to denote the ball centered at a point $c$ with radius $r>0$. The problem of Definition~\ref{def-outlier} in fact is to find $k$ equal-sized balls, $\{\mathbb{B}(c_1, r), \cdots, \mathbb{B}(c_k, r)\}$, to cover at least $(1-\gamma)n$ points of $P$, and the radius $r$ is minimized. 
%If $k=1$, the problem is also called MEB with outliers. 
%\end{definition}

%\begin{remark}
Obviously, in the above definitions, if we set $\gamma=0$ ($\gamma_1= \gamma_2=0$), the problems should be equivalent with the ordinary one-class (two-class) SVM and $k$-center clustering. 
%\end{remark}

\textbf{Approximation ratio.} We also need to clarify the ``approximation ratio'' for a solution. Let $\lambda\geq1$. For the SVM with outliers problems, denote by $w_{opt}$   the width of the optimal margin. Then, for any solution having the margin width $w\geq \frac{1}{\lambda} w_{opt}$, we say that it is a $\frac{1}{\lambda}$-approximate solution. For $k$-center clustering with outliers, denote by $r_{opt}$ the radius of the optimum solution. Then, for any solution having radius $r\leq \lambda r_{opt}$, we say that it is a $\lambda$-approximate solution. 

\vspace{-0.1in}

\subsection{Sparse Approximation Algorithms}
\label{sec-mgalg}
\vspace{-0.1in}

%We briefly introduce  the MEB algorithm~\cite{badoiu2003smaller} (we call it BC's algorithm) and Gilbert algorithm~\cite{gilbert1966iterative} in this section. 
We  introduce two sparse approximation algorithms which will be used as the sub-routines in our framework.  
%Gilbert's algorithm~\citep{gilbert1966iterative} is to compute the shortest distance from the origin to the convex hull of a given set of points. 
%\cite{badoiu2003smaller} proposed an elegant core-set algorithm for finding the center of the minimum enclosing ball for a given set of points (we call it BC's algorithm for short). 
%In fact, these two algorithms both fall under the umbrella of the {\em Frank-Wolfe} method~\cite{frank1956algorithm}, which has been systematically studied by~\cite{C10}. 

\begin{figure}[]
    \centering
  \includegraphics[height=1in]{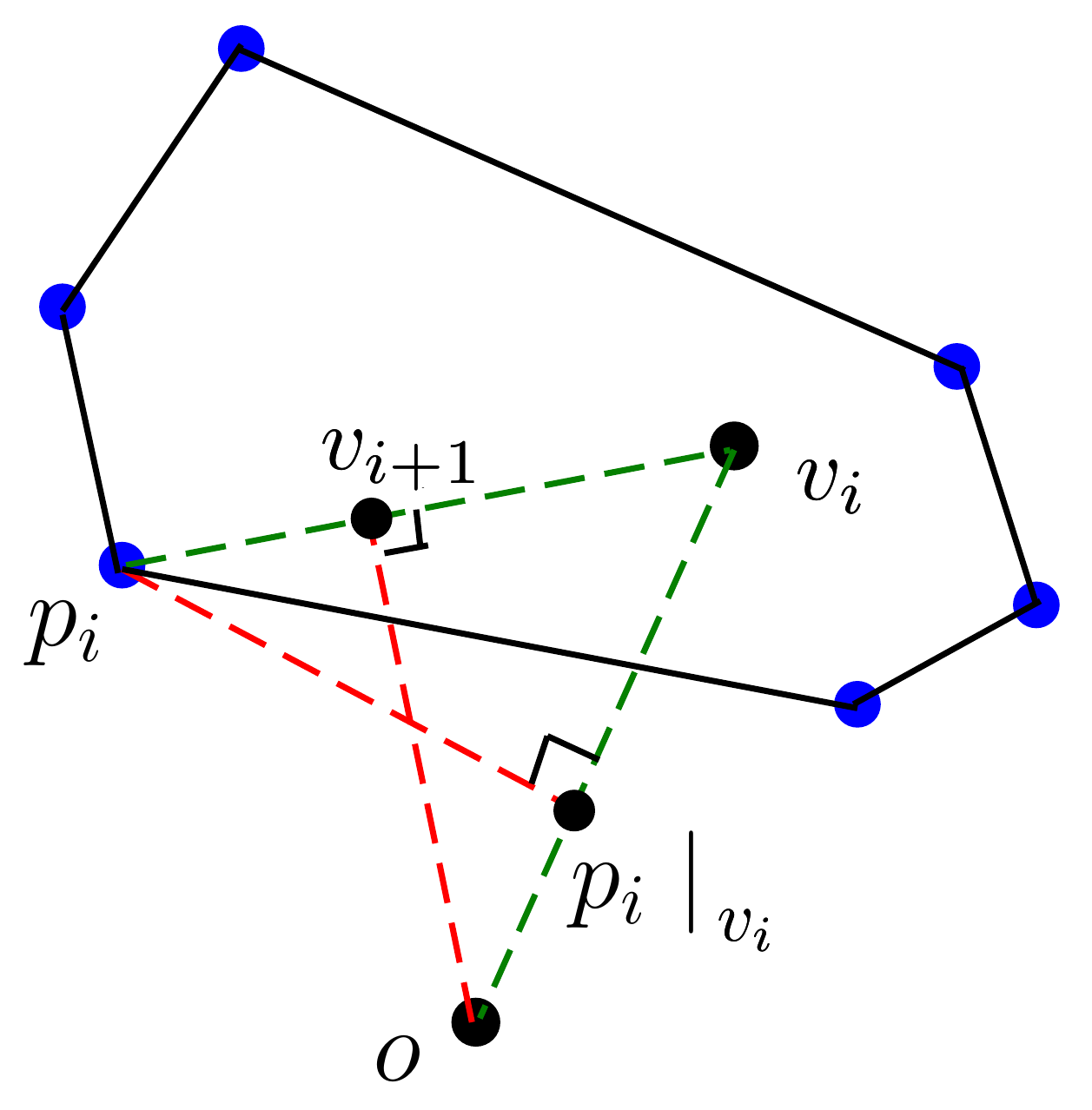}
       \caption{An illustration of step 2 in Algorithm~\ref{alg-gilbert}. After finding the point $p_i$, the algorithm updates $v_i$ to $v_{i+1}$.}
  \label{fig-gilbert}
  \vspace{-0.2in}
 \end{figure}

\textbf{Gilbert's algorithm~\citep{gilbert1966iterative}.} The algorithm is designed for solving the following  {\em polytope distance problem}.  
%\vspace{0.05in}
Assume $o$ is the origin and let $S$ be a given set of points in $\mathbb{R}^d$. The problem is to find a point $q$ inside the convex hull of $S$ (denoted as $\mathtt{conv}(S)$) such that the norm $||q||$ is minimized. The obtained norm,   $\min_{q\in \mathtt{conv}(S)}||q||$, is called the polytope distance between $o$ and $S$.

%\vspace{0.05in}

Before presenting the algorithm, we introduce the following definition first. For any two points $p$ and $q\in\mathbb{R}^d$, denote by $ p\mid_{q}$ the orthogonal projection of  $p$ on the supporting line of segment  $\overline{oq}$. The norm $||p\mid_{q}||$ is called the ``projection distance'' of $p$ to the origin $o$ on $\overline{oq}$.

\begin{definition}[\textbf{$\epsilon$-Approximation of Polytope Distance}]
\label{def-app-polytope}
Let $q \in \mathtt{conv}(S)$ and $\epsilon\in (0,1)$. The point $q$ is an $\epsilon$-approximation of the polytope distance  if  
$||q||\leq \frac{1}{1-\epsilon}||p\mid_{q}||$ for any $p\in S$. 
\end{definition}
\vspace{-0.05in}

Obviously, $\min_{p\in S}\{||p\mid_{q}||\}$ is no larger than the real polytope distance. Thus, if $q$ is an $\epsilon$-approximation, its norm $||q||$ is no larger than $\frac{1}{1-\epsilon}$ times the real polytope distance. Gilbert's algorithm (Algorithm~\ref{alg-gilbert}) is a standard greedy algorithm that improves the current solution by selecting the point $p_i\in S$ having the smallest projection distance in each iteration, until an $\epsilon$-approximation is achieved. 

\begin{algorithm}[h]
   \caption{Gilbert's Algorithm \citep{gilbert1966iterative}}
   \label{alg-gilbert}
\begin{algorithmic}
   \STATE {\bfseries Input:} A point set $S$ in $\mathbb{R}^d$, the origin $o$, and $\epsilon\in (0,1)$.
    \begin{enumerate}
   \item Select the point $v_1$ that is the closest point in $S$ to $o$.
   \item Initialize $i=1$. Iteratively perform the following steps until $v_i$ is an $\epsilon$-approximation as Definition~\ref{def-app-polytope}.
   \begin{enumerate}
   \item Find the point $p_i \in S$ who has the smallest projection distance to $o$ on $\overline{ov_{i}}$, {\em i.e.,} $p_i=\arg\min_{p\in S}\{||p\mid_{v_i}||\}$
   (see Figure~\ref{fig-gilbert}).
   \item Let $v_{i+1}$ be the point on segment $\overline{v_i p_i}$ closest to the origin $o$; update $i=i+1$.
    \end{enumerate}
    \end{enumerate}
        \STATE {\bfseries Output:} $v_i$.    
\end{algorithmic}
\end{algorithm}

\begin{proposition}[\citep{GJ09,C10}]
\label{the-gilbert-polytope}
Suppose $D=\max_{p,q\in S}||p-q||$ and $\rho\geq 0$ is the polytope distance between $o$ and $S$, and let $E=D^2/\rho^2$. 
For any $\epsilon\in (0,1)$,  Algorithm~\ref{alg-gilbert} takes at most $2\lceil 2E/\epsilon\rceil$ steps to achieve an $\epsilon$-approximation of polytope distance.
\end{proposition}

\vspace{-0.05in}

The selected set of points $T=\{p_i\mid 1\leq i\leq 2\lceil 2E/\epsilon\rceil\}$ in Algorithm~\ref{alg-gilbert} is called the ``coreset'' of polytope distance. 

\textbf{BC's Algorithm\citep{badoiu2003smaller}.}  Let $0<\epsilon<1$ and $S$ be a set of points in $\mathbb{R}^d$. The BC's algorithm is an iterative procedure for computing an approximate  center of $\mathtt{MEB}(S)$ (we use $\mathtt{MEB}(\cdot)$ to denote the minimum enclosing ball of a given point set). Initially, it selects an arbitrary point from $S$ and places it into an initially empty set $T$. 
In each of the following $\lceil 2/\epsilon\rceil$ iterations, 
the algorithm updates the center of $\mathtt{MEB}(T)$ and adds to $T$ the farthest point from the current center of $\mathtt{MEB}(T)$. 
Finally, the center of $\mathtt{MEB}(T)$ yields a $(1+\epsilon)$-approximation for $\mathtt{MEB}(S)$. The selected set of $\lceil 2/\epsilon\rceil$ points ({\em i.e.}, $T$) is called the ``coreset'' of MEB.

\begin{remark}\textbf{(sparse approximation)}
\label{rem-represent}
A benefit of these two algorithms is that the solution, the point $v_i$ inside $\mathtt{conv}(S)$ or the approximate ball center of $S$, can always be represented by a convex combination of the obtained coreset $T$ (which in fact is a sparse approximation of the optimal solution). Moreover, the coefficients of this convex combination have been simultaneously generated after running the algorithm (since these coefficients are always recorded in each iteration). In our following algorithms, we will use these coefficients to recover the solution from the JL transform. 
\end{remark}

\section{SVM with Outliers}
\label{sec-svm}
\vspace{-0.1in}
To solve the SVM with outliers problems, we introduce an important geometric result below. 
Let $\Delta o y_1 y_2$ be a right triangle in a two-dimensional plane, where $o$ is the origin, $y_1=(a, 0)$, and $y_2=(a, b)$ with $a, b>0$ (see Figure~\ref{fig-tri}). Suppose the points $y_1$ and $y_2$ are moved to the new locations $y'_1=(a_0, 0)$ and $y'_2=(a', b')$, respectively. Consequently, they form a new triangle $\Delta o y'_1 y'_2$.

\begin{lemma}\textbf{(triangle preservation)}
	\label{lem-tri}
	Let $\delta$ be a small positive value. If the three sides of the new triangle $\Delta o y'_1 y'_2$ satisfy 
	
	\begin{equation}
		\left.\begin{aligned}
			||y'_1-o||^2&\in ||y_1-o||^2\pm \delta,\\
			||y'_2-o||^2&\in ||y_2-o||^2\pm \delta,\\
			||y'_1-y'_2||^2&\in ||y_1-y_2||^2\pm \delta,
		\end{aligned}
		\hspace{0.3in} \right  \} \label{for-tri-1}
		\qquad  
	\end{equation}
	then we have 
	$a'\geq \frac{2a^2-3\delta}{2\sqrt{a^2+\delta}}$.
\end{lemma}
\begin{proof}
	
	\begin{figure}[]
		\centering
		\includegraphics[height=1.2in]{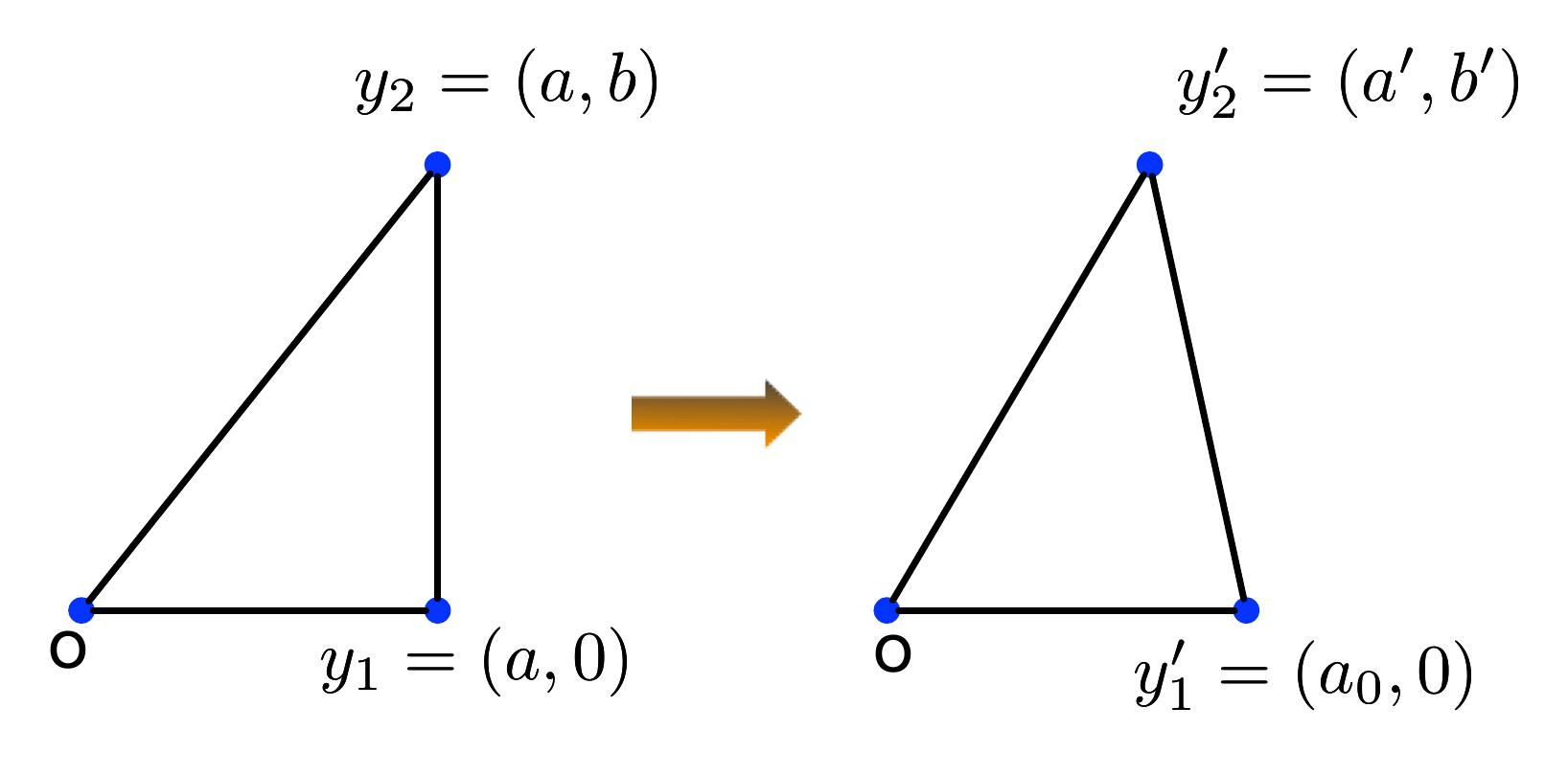}
		\vspace{-0.2in}
		\caption{An illustration for Lemma~\ref{lem-tri}.}
		\label{fig-tri}
		\vspace{-0.15in}
	\end{figure}
	
	From (\ref{for-tri-1}), we directly know
	\begin{eqnarray}
		a^2_0&\in& a^2\pm \delta;\label{for-tri-2}\\
		(a'-a_0)^2+(b')^2&\in&  b^2\pm \delta;\label{for-tri-3}\\
		(a')^2+(b')^2&\in& (a^2+b^2)\pm \delta.\label{for-tri-4}
	\end{eqnarray}
	Through (\ref{for-tri-3}) and (\ref{for-tri-4}), we obtain
	\begin{eqnarray}
		(a^2+b^2)-\delta&\leq&(a')^2+(b')^2\nonumber\\
		&=&(a'-a_0)^2+(b')^2+2a_0 a'-(a_0)^2\nonumber\\
		&\leq &  b^2+\delta+2a_0 a'-(a_0)^2. \nonumber\\
		\implies a^2+a_0^2&\leq&2a_0 a'+2\delta. \label{for-tri-5}
	\end{eqnarray}
	Further, we have
	\begin{eqnarray}
		2a^2-3\delta &\underbrace{\leq}_{\text{by (\ref{for-tri-2})}} &a^2-2\delta+a_0^2\nonumber\\
		&\underbrace{\leq}_{\text{by (\ref{for-tri-5})}} &2a_0 a'\underbrace{\leq}_{\text{by (\ref{for-tri-2})}} 2\sqrt{a^2+\delta} a'. \label{for-tri-6}
	\end{eqnarray}
	Finally, we have the inequality $a'\geq \frac{2a^2-3\delta}{2\sqrt{a^2+\delta}}$ from (\ref{for-tri-6}).
	%So we complete the proof.
\end{proof}

\vspace{-0.1in}
\subsection{One-class SVM with Outliers}
\label{sec-oneclass}
\vspace{-0.05in}
To present our result, we need to relate polytope distance to one-class SVM through the following proposition. 
Let $S\subset \mathbb{R}^d$ and $x$ be a point in $\mathbb{R}^d$.

\begin{proposition}[\citep{GJ09,C10}]
	\label{lem-margin1}
	\textbf{(\rmnum{1})} If $x$ is the optimal solution of the polytope distance from $o$ to $S$, then $x$ also corresponds to the optimal solution for the one-class SVM on $S$, {\em i.e.,} the hyperplane $\mathcal{H}_x$ yields the maximum margin separating the origin and $S$.
	
	\textbf{(\rmnum{2})}  Suppose $x$ is an $\epsilon$-approximation of the polytope distance from the origin $o$ to $S$ (see Definition~\ref{def-app-polytope}). Let $y=(1-\epsilon) x$. Then the hyperplane $\mathcal{H}_{y}$ separates $o$ and $S$, and its width, {\em i.e.,} $||y||$, is at least $(1-\epsilon)$ times the maximum width.
\end{proposition}

%\begin{lemma}[\cite{GJ09}]
%	\label{lem-margin2}
%	If $x$ is an $\epsilon$-approximation of the polytope distance from the origin $o$ to $S$ (see Definition~\ref{def-app-polytope}). Let $y=(1-\epsilon) x$. Then 
%	the hyperplane $\mathcal{H}_{y}$ separates $o$ and $S$, and its width, {\em i.e.,} the norm $||y||$, is at least $(1-\epsilon)$ times the maximum width.
%\end{lemma}

Proposition~\ref{lem-margin1} \textbf{(\rmnum{2})} suggests that we can find an approximate solution for one-class SVM through Gilbert's Algorithm. We further show how to apply it to the case with outliers. Given an instance $(P, \gamma)$ of one-class SVM with outliers, we suppose that $x_{opt}$ is the point yielding the optimal separating margin for the instance $(P, \gamma)$, and suppose $P_{opt}$ is the set of induced inliers; that is, the hyperplane $\mathcal{H}_{x_{opt}}$ separates the origin $o$ and $P_{opt}$. Similar with the notations defined in Proposition~\ref{the-gilbert-polytope}, we let $D=\max_{p, q\in P}||p-q||$, $\rho=||x_{opt}||$, and $E=D^2/\rho^2$. 

Our algorithm (Algorithm~\ref{alg-oneclass}) follows the generic framework introduced in Section~\ref{sec-our}. The high-level idea is as follows. First, we apply the JL transform $f$ to $P$ and solve the new instance $(f(P), \gamma)$ by any existing algorithm for one-class SVM with outliers. Suppose the obtained solution is $\bar{x}$ and it can be represented by a convex combination of a subset of $f(P)$ through Gilbert's algorithm. Then, we transform the solution $\bar{x}$ back to the original space $\mathbb{R}^d$ by using the coefficients of the convex combination (see Remark~\ref{rem-represent}). 

\begin{algorithm}[tb]
	\caption{\sc{ One-class SVM with outliers }}
	\label{alg-oneclass}
	\begin{algorithmic}
		\STATE {\bfseries Input:} A set $P$ of $n$ points in $\mathbb{R}^d$, $\epsilon, \epsilon_0>0$, and $\gamma\in (0,1)$.
		\begin{enumerate}
			\item Apply the JL transform $f$ to reduce the dimensionality $d$ to be $O(\frac{1}{\epsilon^2}\log n)$, and run any existing algorithm $\mathcal{A}$ for SVM with outliers on the new instance $(f(P), \gamma)$. 
			\item Suppose the normal vector returned by $\mathcal{A}$ is $v$. Then, we project all the points $f(P)$ to the vector $v$, and find the margin along $v$ such that exactly $(1-\gamma)n$ points of $f(P)$ are separated from the origin. These $(1-\gamma)n$ points form the set of inliers $f(S)$. 
			
			\item Compute an $\epsilon_0$-approximate solution for the polytope distance problem of $f(S)$ through Gilbert's algorithm, and denote by $\bar{x}$ the obtained point in $\mathtt{conv}(f(S))$. 
			\begin{itemize}
				\item In Gilbert's algorithm, $\bar{x}$ is represented as a convex combination of the points of $f(S)$, say $\bar{x}=\sum_{q\in f(S)}\alpha_q q$ where $\sum_{q\in f(S)}\alpha_q=1$ and $\alpha_q\geq 0$ for $\forall q\in f(S)$. 
				\item Let $f^{-1}(\bar{x})=\sum_{q\in f(S)}\alpha_q f^{-1}(q)$ (see Remark~\ref{rem-recovery} for the explanation on $f^{-1}$).
			\end{itemize}
		\end{enumerate}
		\STATE {\bfseries Output:}  $f^{-1}(\bar{x})$ as the solution. 
	\end{algorithmic}
\end{algorithm}

\begin{theorem}\textbf{(margin preservation)}
	\label{the-oneclass}
	We set $\epsilon_0\in (0,1)$ and  $\epsilon=\frac{1}{5}\frac{\epsilon_0}{E+1}$ in Algorithm~\ref{alg-oneclass}. 
	Suppose $\lambda\geq 1$ and the algorithm $\mathcal{A}$ used in Step 1 yields a $1/\lambda$-approximate solution of one-class SVM with outliers. The returned vector $f^{-1}(\bar{x})$ of Algorithm~\ref{alg-oneclass} yields a $\frac{1}{\lambda}(1-\epsilon_0)^3$-approximate solution for the instance $(P, \gamma)$  with constant probability\footnote{The ``constant probability'' directly comes from the success probability of JL transform. Let $\eta\in (0,1)$. If we set the reduced dimensionality $\tilde{d}=O(\frac{1}{\epsilon^2}\log\frac{n}{\eta})$, the success probability will be $1-\eta$~\citep{dasgupta2003elementary}. For simplicity, in this paper we always assume $\eta$ is a fixed small constant and say that the JL transform achieves a constant success probability.}.
\end{theorem}
\begin{remark}\textbf{(the recovery step)}
\label{rem-recovery}
	The mapping $f$ is not a bijection, so we cannot directly define $f^{-1}$ for an arbitrary point in the space. However we can view $f$ as a bijection if restricting the domain to be the input set: from $P$ to $f(P)$. So $f^{-1}$ can be well defined for the points in $f(P)$ ({\em i.e.,} $f^{-1}\big(f(p)\big)=p$ for any $p\in P$). For an arbitrary point $q$ in the dimensionality-reduced space, if it can be represented as a convex combination of the points in $f(P)$, we can define $f^{-1}(q)$ as well based on the coefficients because $f$ is linear ({\em e.g.}, if $q=\frac{1}{3}f(p_1)+\frac{2}{3}f(p_2)$ with $p_1$ and $p_2\in P$, then $f^{-1}(q)=f^{-1}\big(\frac{1}{3}f(p_1)+\frac{2}{3}f(p_2)\big)=\frac{1}{3}p_1+\frac{2}{3}p_2$). By this way, we can define $f^{-1}(\bar{x})$ as the final output of Algorithm~\ref{alg-oneclass}.
\end{remark}

Before proving Theorem~\ref{the-oneclass}, we first introduce the following lemma that was proved in~\citep{DBLP:conf/compgeom/AgarwalHY07,DBLP:conf/compgeom/Sheehy14}. The result can  be viewed as an extension of the original JL Lemma (Lemma~\ref{lem-jl}). A key of Lemma~\ref{lem-radius} is that it relaxes the error bound to be a unified $\epsilon\hspace{0.02in} \mathtt{rad}(S)^2$ (we use $\mathtt{rad}(S)$ to denote the radius of the minimum enclosing ball of $S$), rather than the bound $\epsilon||p-q||^2$ in the inequality (\ref{for-jl}). 
%A similar result of (\ref{for-radius-1}) was also proved by \cite{DBLP:conf/compgeom/AgarwalHY07}.

\begin{lemma}
	\label{lem-radius}
	Let $f$ be the JL transform of $P$ from $\mathbb{R}^d$ to $\mathbb{R}^{O(\frac{1}{\epsilon^2}\log n)}$. With constant probability, for any subset $S\subseteq P$, any point $p\in S$,  and any point $q$ inside $\mathtt{conv}(S)$, we have  
%	(\rmnum{1}) 
%	\begin{eqnarray}
%		\mathtt{rad}(f(S))^2\in (1\pm \epsilon)\mathtt{rad}(S)^2; \label{for-radius-1}
%	\end{eqnarray}
%	(\rmnum{2}) for any point $p\in S$ and any point $q$ inside the convex hull of $S$,  
	\begin{eqnarray}
		\big|||p-q||^2-||f(p)-f(q)||^2\big|\leq \epsilon\hspace{0.05in} \mathtt{rad}(S)^2.\label{for-radius-2}
	\end{eqnarray}
\end{lemma}

	\begin{figure}[]
		\centering
		\includegraphics[height=1in]{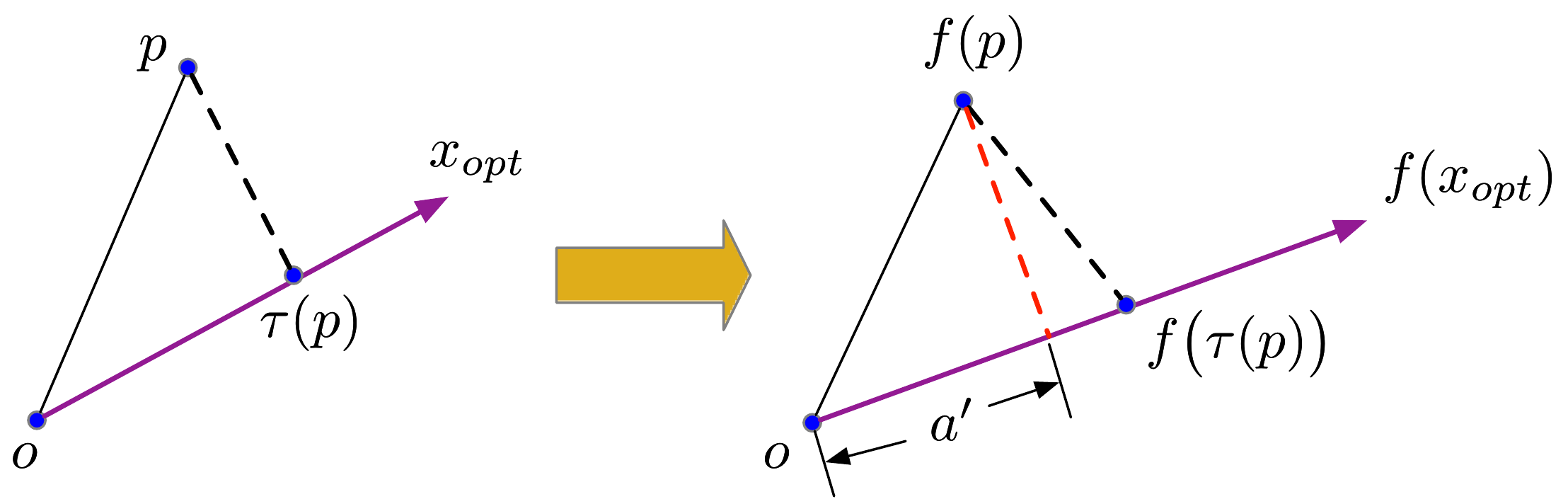}
%		\vspace{-0.05in}
		\caption{Because the JL transform $f$ is linear mapping, the point $f\big(\tau(p)\big)$ is located on the line determined by the vector $f(v_{opt})$. The point $f(p)$ has projection distance $a'$ to the origin $o$.}
		\label{fig-svm-proj}
		\vspace{-0.15in}
	\end{figure}

\begin{proof}\textbf{(of Theorem~\ref{the-oneclass})}
		For any point $p\in P_{opt}$, denote by $\tau(p)$ the projection of $p$ on the vector $x_{opt}$. Then, we focus on the right triangle $\Delta o \tau(p) p$ and its image, $\Delta o f\big(\tau(p)\big) f(p)$, in the lower dimensional space induced by the JL transform $f$ ({\em w.l.o.g.}, we assume the two spaces share the same origin $o$). See Figure~\ref{fig-svm-proj}. From Proposition~\ref{lem-margin1}, we know that $x_{opt}$ should be inside $\mathtt{conv}(P_{opt})$. Consequently, the triangle $\Delta o \tau(p) p$  is inside the convex hull of $\{o\}\cup P_{opt}$. Moreover, the set $\{o\}\cup P_{opt}$ is covered by a ball with radius no larger than $ D+\rho$. Therefore, through Lemma~\ref{lem-radius}, we have 
	\begin{eqnarray}
		||f\big(\tau(p)\big)-o||^2&\in& ||\tau(p)-o||^2\pm \epsilon (D+\rho)^2;\nonumber\\
		|| f(p)-o||^2&\in& ||p-o||^2\pm \epsilon (D+\rho)^2;\nonumber\\
		||f\big(\tau(p)\big)- f(p)||^2&\in& ||\tau(p)-p||^2\pm \epsilon (D+\rho)^2. \nonumber
	\end{eqnarray}
	We apply Lemma~\ref{lem-tri} to $\Delta o \tau(p) p$ by letting $y_1=\tau(p)$, $y_2=p$, $y'_1=f\big(\tau(p)\big)$, $y'_2=f(p)$, and $\delta= \epsilon (D+\rho)^2$. For convenience, we use the same notations as Lemma~\ref{lem-tri} and have
	\begin{eqnarray}
		a'\geq \frac{2a^2-3\delta}{2\sqrt{a^2+\delta}}=\sqrt{a^2+\delta}-\frac{5\delta}{2\sqrt{a^2+\delta}}.\label{for-oneclass-2}
	\end{eqnarray}
	Note that $a=||\tau(p)||\geq \rho$ and $\epsilon=\frac{1}{5}\frac{\epsilon_0}{E+1}$ with $E=D^2/\rho^2$; then we have
	\begin{eqnarray}
		a'&\geq& \sqrt{\rho^2+\delta}-\frac{5\delta}{2\sqrt{\rho^2+\delta}}\nonumber\\
		&\geq& \rho-\frac{5}{2}\frac{\delta}{\rho}\nonumber\\
		&\underbrace{=}_{\text{plug the values of $\delta$ and $\epsilon$}}&\rho-\frac{\epsilon_0}{2}\cdot\frac{(D+\rho)^2}{E+1}\cdot\frac{1}{\rho}\nonumber\\
		&\geq&\rho-\frac{\epsilon_0}{2}\cdot\frac{2D^2+2\rho^2}{E+1}\cdot\frac{1}{\rho}\nonumber\\
		&=&(1-\epsilon_0)\rho.\label{for-oneclass-3}
	\end{eqnarray} 
	(\ref{for-oneclass-3}) implies that for any point $p\in P_{opt}$, $f(p)$'s projection on the vector $f(x_{opt})$ has distance $\geq (1-\epsilon_0)\rho$ to the origin. 
	That is, the vector $ f(x_{opt}) $ yields a solution for the instance $\big(f(P), \gamma\big)$ with separating margin width at least $(1-\epsilon_0)\rho$. Further, 
	because $v$ yields a $\frac{1}{\lambda}$-approximate solution for $\big(f(P), \gamma\big)$ (step 2 of Algorithm~\ref{alg-oneclass}), we have 
	\begin{eqnarray}
		w_v\geq \frac{1}{\lambda}(1-\epsilon_0)\rho \label{for-oneclass-4}
	\end{eqnarray}
	where $w_v$ denotes the margin width induced by $v$. Also, since $\bar{x}$ is an $\epsilon_0$-approximation of the polytope distance problem for the instance $f(S)$, through Proposition~\ref{lem-margin1} we know that  the vector $\bar{x}$ yields a margin separating  the origin  and $f(S)$ with the width 
	\begin{eqnarray}
		||(1-\epsilon_0)\bar{x}||\geq (1-\epsilon_0)w_v\geq \frac{1}{\lambda}(1-\epsilon_0)^2\rho. \label{for-oneclass-7}
	\end{eqnarray}
	
	Then, we consider the inverse mapping $f^{-1}$ from $\mathbb{R}^{O(\frac{1}{\epsilon^2}\log n)}$ to $\mathbb{R}^d$. For any $q\in f(S)$, consider the right triangle $\Delta o \tau(q) q$ and its image, $\Delta o f^{-1}(\tau(q)) f^{-1}(q)$ in $\mathbb{R}^d$, where $\tau(q)$ is the projection of $q$ on the vector $\bar{x}$. Since $\tau(q)\in \mathtt{conv}\big(\{o\}\cup f(S)\big)$, it can be represented as a convex combination of $\{o\}\cup f(S)$; therefore, we can define its image $f^{-1}(\tau(q))$ as the convex combination of $\{o\}\cup S$ (see Remark~\ref{rem-recovery}). 
	So we can apply the same manner for proving the above (\ref{for-oneclass-3}) with replacing $f$ and $P_{opt}$ by $f^{-1}$ and $f(S)$ respectively. 
	For any point $q\in f(S)$, $f^{-1}(q)$'s projection on the vector $f^{-1}(\bar{x})$ has distance at least 
	\begin{eqnarray}
		(1-\epsilon_0)||(1-\epsilon_0)\bar{x}||\geq \frac{1}{\lambda}(1-\epsilon_0)^3\rho \label{for-oneclass-6}
	\end{eqnarray}
	to the origin by (\ref{for-oneclass-7}). In other words, if we let $S$ be the set of inliers, the vector 
	$f^{-1}(\bar{x})$ yields a $\frac{1}{\lambda}(1-\epsilon_0)^3$-approximation for the instance $(P, \gamma)$.
\end{proof}

\textbf{Time complexity.} Suppose the black box algorithm $\mathcal{A}$ of Algorithm~\ref{alg-oneclass} has the time complexity $\Gamma(n, d)$ for an instance of $n$ points in $d$-dimensional space. We reduce the dimensionality from $d$ to $\tilde{d}=O(\frac{1}{\epsilon^2}\log n)=O(\frac{E^2}{\epsilon_0^2}\log n)$ (since we require $\epsilon=\frac{1}{5}\frac{\epsilon_0}{E+1}$). Obviously, the total time complexity of Algorithm~\ref{alg-oneclass} consists of three parts: 
\begin{eqnarray}
\mathtt{Time}_{JL}+\Gamma(n, \tilde{d})+\mathtt{Time}_{rec},\label{for-time-one-class}
\end{eqnarray}
where $\mathtt{Time}_{JL}$ indicates the time complexity of JL transform, and $\mathtt{Time}_{rec}$ indicates the time complexity of the recovery step (Step 3). For JL transform, if we simply use the random Gaussian matrix~\citep{dasgupta2003elementary}, $\mathtt{Time}_{JL}$ will be $O(n\cdot d\cdot \tilde{d})$; the time complexity can be further improved by using the techniques like~\citep{ailon2009fast}. For the recovery step, we run Gilbert's algorithm in the dimensionality-reduced space and compute $f^{-1}(\bar{x})$ based on the coefficients in the original space; therefore, according to Proposition~\ref{the-gilbert-polytope}, $\mathtt{Time}_{rec}=O\big(\frac{E}{\epsilon_0}\cdot(n\cdot\tilde{d}+d)\big)$ (there are only $O(\frac{E}{\epsilon_0})$ non-zero coefficients). It is worth noting that $\mathtt{Time}_{JL}$ and $\mathtt{Time}_{rec}$ are usually much lower than the second term $\Gamma(n, \tilde{d})$. For example, most of the existing robust SVM algorithms (see Section~\ref{sec-relate}) take the time complexity at least quadratic in the number of input points $n$.

\subsection{Two-class SVM with outliers}

 Consider the connection between 
polytope distance and two-class SVM.
Let $Q_{1}$ and $Q_{2}$ be two point sets in $\mathbb{R}^d$.  The \textbf{Minkowski Difference} $\mathtt{MD}(Q_{1},Q_{2})$ is the set of all difference vectors from  $\mathtt{conv}(Q_1)$ and $\mathtt{conv}(Q_2)$, {\em i.e.,} 
$\mathtt{MD}(Q_1, Q_2) =\{u-v\mid u\in \mathtt{conv}(Q_1), v\in \mathtt{conv}(Q_2)\} $. Note that $\mathtt{MD}(Q_1, Q_2)$ is also a convex polytope. \cite{GJ09} showed that 
%\begin{lemma}
%\label{the-tpolytope}[\citep{GJ09}]
finding the shortest distance between two polytopes $\mathtt{conv}(Q_1)$ and $\mathtt{conv}(Q_2)$ is equivalent to finding the polytope distance from the origin to $\mathtt{MD}(Q_1,Q_2)$.
%\end{lemma}
In other words, to find the maximum margin separating two point sets $Q_1$ and $Q_2$, we only need to find the maximum margin separating the origin and $\mathtt{MD}(Q_1, Q_2)$.  However, directly computing $\mathtt{MD}(Q_1,Q_2)$ takes quadratic time $O(|Q_1||Q_2| d)$. Actually, we do not need to explicitly compute $\mathtt{MD}(Q_1, Q_2)$ for obtaining its polytope distance to the origin. In each iteration of the Gilbert's algorithm, we just select the point $p_i$ that has the closest projection to $o$. In $\mathtt{MD}(Q_1,Q_2)$, this point ``$p_i$'' should be the difference vector $q_i-q'_i$, where $q_i\in Q_1$ and $q'_i\in Q_2$, and 
$$p_i\mid_{v_i}=q_i\mid_{v_i}-q'_i\mid_{v_i}.$$
Therefore, 
$q_i$ should be the point having the closest projection to $o$ from $Q_1$ and $q'_i$ should be the point having the farthest projection to $o$ from $Q_2$. As a consequence, we just need to select the points $q_i$ and $q'_i$ in each iteration, and thus the complexity for computing the polytope distance is still linear. The following theorem can be proved through the similar idea for the one-class case in Section~\ref{sec-oneclass}. 
\begin{theorem}
\label{the-twoclass}
Let $\epsilon_0\in (0,1)$ and  $\epsilon=\frac{1}{5}\frac{\epsilon_0}{E+1}$. 
Suppose $\lambda\geq 1$ and the algorithm $\mathcal{A}$ used in Step 1 of Algorithm~\ref{alg-twoclass} yields a $1/\lambda$-approximate solution of SVM with outliers. The returned vector $f^{-1}(\bar{x})$ of Algorithm~\ref{alg-twoclass} yields a $\frac{1}{\lambda}(1-\epsilon_0)^3$-approximate solution for the instance $\big(P_1, P_2, \gamma_1, \gamma_2\big)$  with constant probability.
\end{theorem}
%Given a point $x\in\mathbb{R}^d$, we use $\mathcal{M}_x$ to denote the margin that is orthogonal to $x$ and separates $Q_1$ and $Q_2$. 
%
%\begin{lemma} 
%\label{lem-margin3}
%Let $x\in \mathtt{MD}(Q_1,Q_2)$ be an $\epsilon$-approximation of the polytope distance  yielded by  Gilbert Algorithm  for some small constant $\epsilon\in (0,1)$. Then 
% $\mathcal{M}_{x}$ separates $Q_1$ and $Q_2$, and its width is at least $(1-\epsilon)$ times that of the maximum width.
%\end{lemma}
%
%
%Finally, we have the algorithm for solving the two-class SVM with outliers problem. 

\begin{algorithm}[h]
   \caption{\sc{ Two-class SVM with outliers }}
   \label{alg-twoclass}
\begin{algorithmic}
   \STATE {\bfseries Input:} Two point sets $P_1$ and $P_2$  in $\mathbb{R}^d$ with $n=|P_1\cup P_2|$, $\epsilon, \epsilon_0>0$, and $\gamma_1, \gamma_2\in (0,1)$. 
      \begin{enumerate}
\item Apply the JL transform $f$ to reduce the dimensionality $d$ to be $O(\frac{1}{\epsilon^2}\log n)$, and run any existing algorithm $\mathcal{A}$ for SVM with outliers on the new instance $\big(f(P_1), f(P_2), \gamma_1, \gamma_2\big)$. 

\item Suppose the normal vector returned by $\mathcal{A}$ is $v$. Then, we project all the points of $f(P_1)$ and $f(P_2)$ to the vector $v$, and find the margin along $v$ such that exactly $(1-\gamma_1)|P_1|$ points of $f(P_1)$ and $(1-\gamma_2)|P_2|$ points of $f(P_2)$ are separated. These points form the sets of inliers $f(S_1)$ and $f(S_2)$, respectively.

\item Compute an $\epsilon_0$-approximate solution for the polytope distance between $f(S_1)$ and $f(S_2)$ through Gilbert's algorithm, and denote by $\bar{x}$ the obtained point in $\mathtt{MD}(f(S_1), f(S_2))$. 
\begin{itemize}
\item In Gilbert's algorithm, $\bar{x}$ is represented as a convex combination of the points of $f(S)$ where $S=S_1\cup S_2$, say $\bar{x}=\sum_{q\in f(S)}\alpha_q q$ where $\sum_{q\in f(S)}\alpha_q=1$ and $\alpha_q\geq 0$ for $\forall q\in f(S)$. 
\item Let $f^{-1}(\bar{x})=\sum_{q\in f(S)}\alpha_q f^{-1}(q)$.
\end{itemize}
\end{enumerate}
    \STATE {\bfseries Output:}  $f^{-1}(\bar{x})$ as the final solution. 
\end{algorithmic}
\end{algorithm}

\vspace{-0.1in}

\section{$k$-Center Clustering with Outliers}
\label{sec-meb}
\vspace{-0.1in}

 In this section, we also follow the generic framework introduced in Section~\ref{sec-our} and present our algorithm for $k$-center clustering with outliers (Algorithm~\ref{alg-kcenter}).  
To explain our idea more clearly, 
we first  consider the minimum enclosing ball (MEB) with outliers problem ({\em i.e.,} the case $k=1$), and then extend the result to the general $k$-center clustering with outliers problem. 

\begin{algorithm}[tb]
   \caption{\sc{ $k$-center clustering with outliers }}
   \label{alg-kcenter}
\begin{algorithmic}
   \STATE {\bfseries Input:} A set $P$ of $n$ points in $\mathbb{R}^d$, $k\in \mathbb{Z}^+$, and $\gamma, \epsilon\in (0,1)$.
      \begin{enumerate}
\item Apply the JL transform $f$ to reduce the dimensionality $d$ to be $O(\frac{1}{\epsilon^2}\log n)$, and run any existing algorithm $\mathcal{A}$ for $k$-center clustering with outliers on the new instance $(f(P), \gamma)$. 
\item Let $f(C_1), f(C_2), \cdots, f(C_k)$ be the $k$ clusters of $f(P)$ obtained by the algorithm $\mathcal{A}$. $\sum^k_{j=1}|f(C_j)|=(1-\gamma)n$. 

\item Compute a $(1+\epsilon)$-approximate MEB  for each $f(C_j)$ through  BC's algorithm, and denote by $\bar{c}_j$ the obtained ball center. 
\begin{itemize}
\item Each center $\bar{c}_j$ is represented as a convex combination of the points of $f(C_j)$, say $\bar{c}_j=\sum_{q\in f(C_j)}\alpha_q q$ where $\sum_{q\in f(C_j)}\alpha_q=1$ and $\alpha_q\geq 0$ for $\forall q\in f(C_j)$. 
\item Let the $k$ points $f^{-1}(\bar{c}_j)=\sum_{q\in f(C_j)}\alpha_q f^{-1}(q)$ for $1\leq j\leq k$ (see Remark~\ref{rem-recovery} for the explanation on $f^{-1}$).
\end{itemize}
\end{enumerate}
    \STATE {\bfseries Output:}  $f^{-1}(\bar{c}_j)$, $1\leq j\leq k$, as the $k$ cluster centers. 
\end{algorithmic}
\end{algorithm}

\vspace{-0.1in}

%
%\subsection{MEB with Outliers}
%
%\vspace{-0.05in}

%See the details in Algorithm~\ref{alg-meb}. 

%
%\begin{algorithm}[tb]
%   \caption{\sc{ MEB with outliers }}
%   \label{alg-meb}
%\begin{algorithmic}
%   \STATE {\bfseries Input:} A set $P$ of $n$ points in $\mathbb{R}^d$, and $\gamma, \epsilon\in (0,1)$.
%      \begin{enumerate}
%\item Apply the JL transform $f$ to reduce the dimensionality $d$ to be $O(\frac{1}{\epsilon^2}\log n)$, and run any existing algorithm $\mathcal{A}$ for MEB with outliers on the new instance $(f(P), \gamma)$. 
%\item Let $f(S)$ be the set of $(1-\gamma)n$ inliers of $f(P)$ induced by the algorithm $\mathcal{A}$. 
%
%\item Compute a $(1+\epsilon)$-approximate MEB of $f(S)$ through  BC's algorithm in Section~\ref{sec-mgalg}, and denote by $\bar{c}$ the obtained ball center. 
%\begin{itemize}
%
%\item The center $\bar{c}$ is represented as a convex combination of the points of $f(S)$, say $\bar{c}=\sum_{q\in f(S)}\alpha_q q$ where $\sum_{q\in f(S)}\alpha_q=1$ and $\alpha_q\geq 0$ for $\forall q\in f(S)$. 
%\end{itemize}
%\end{enumerate}
%    \STATE {\bfseries Output:} $f^{-1}(\bar{c})=\sum_{q\in f(S)}\alpha_q f^{-1}(q)$ as the ball center. 
%\end{algorithmic}
%\end{algorithm}

%\vspace{-0.15in}
\begin{theorem}\textbf{(radius preservation)}
\label{the-jl}
Let $k=1$ and $\lambda\geq 1$. Suppose the algorithm $\mathcal{A}$ used in Step 1 of Algorithm~\ref{alg-kcenter} yields a $\lambda$-approximate solution of MEB with outliers. The returned point $f^{-1}(\bar{c}_1)$ of Algorithm~\ref{alg-kcenter} yields a $\lambda\sqrt{\frac{(1+\epsilon)^3}{1-\epsilon}}$-approximate solution for the original instance $(P,\gamma)$ in $\mathbb{R}^d$ with constant probability. 
\end{theorem}
%\vspace{-0.15in}
\begin{proof}
We use $P_{opt}$ and $r_{opt}$ to denote the optimal subset and radius for the instance $(P, \gamma)$. Let $c_{opt}$ be the center of $\mathtt{MEB}(P_{opt})$. Obviously $c_{opt}$ is inside the convex hull of $P_{opt}$. Thus, using the inequality (\ref{for-radius-2}) in Lemma~\ref{lem-radius}, we have
\begin{eqnarray}
\forall p\in P_{opt}, \hspace{0.1in} ||f(p)-f(c_{opt})||^2&\leq& ||p-c_{opt}||^2+\epsilon r^2_{opt}\nonumber\\
&\leq& (1+\epsilon) r^2_{opt}\nonumber\\
\Longrightarrow\max_{p\in P_{opt}}||f(p)-f(c_{opt})||^2&\leq& (1+\epsilon) r^2_{opt}. \label{for-jl1} 
\end{eqnarray}
That is, the ball $\mathbb{B}\big(f(c_{opt}), \sqrt{1+\epsilon}r_{opt}\big)$ covers the whole set $f(P_{opt})$. Therefore, if we denote the optimal radius of the instance $\big(f(P), \gamma\big)$ as $r'_{opt}$, we have 
\begin{eqnarray}
r'_{opt}\leq \sqrt{1+\epsilon}\cdot r_{opt}. \label{for-jl6}
\end{eqnarray}
Next, we prove the upper bound of the radius induced by $f^{-1}(\bar{c}_1)$. Since 
 $\bar{c}_1$ is inside the convex hull of $f(C_1)$,  the point $f^{-1}(\bar{c}_1)$ should be inside the convex hull of $C_1$. Also, because the JL transform $f$ is linear, 
 \begin{eqnarray}
 f\big(f^{-1}(\bar{c}_1)\big)=\sum_{q\in f(C_1)}\alpha_q f\big(f^{-1}(q)\big)=\bar{c}_1. \label{for-inverse1}
 \end{eqnarray}
 Applying the inequality (\ref{for-radius-2})  again, $\forall p\in C_1$, we have 
\begin{eqnarray}
&&||p-f^{-1}(\bar{c}_1)||^2\nonumber\\
&\underbrace{\leq}_{f\big(f^{-1}(\bar{c}_1)\big)=\bar{c}_1}& ||f(p)-\bar{c}_1||^2+\epsilon \hspace{0.05in} \mathtt{rad}(C_1)^2\nonumber\\
&\leq& \max_{q\in C_1}||f(q)-\bar{c}_1||^2+\epsilon \hspace{0.05in} \mathtt{rad}(C_1)^2\nonumber\\
&\leq& \max_{q\in C_1}||f(q)-\bar{c}_1||^2+\frac{\epsilon}{1-\epsilon} \hspace{0.05in} \mathtt{rad}\big(f(C_1)\big)^2\nonumber\\
&\leq& \frac{1}{1-\epsilon}\max_{q\in C_1}||f(q)-\bar{c}_1||^2,\label{for-jl2}
\end{eqnarray}
where the last inequality comes from the fact $\mathtt{rad}\big(f(C_1)\big)^2\leq \max_{q\in C_1}||f(q)-\bar{c}_1||^2$ (otherwise, we can let $\bar{c}_1$ be the center of $f(C_1)$ and the resulting radius  $ \max_{q\in C_1}||f(q)-\bar{c}_1||^2<\mathtt{rad}\big(f(C_1)\big)^2$, which is a contradiction). Recall that $\bar{c}_1$ is a $(1+\epsilon)$-approximate center of the $\lambda$-approximate solution of the instance $\big(f(P), \gamma\big)$. So 
%\begin{eqnarray}
$\max_{q\in C_1}||f(q)-\bar{c}_1||\leq (1+\epsilon)\cdot\lambda \cdot r'_{opt}$. 
%\end{eqnarray} 
Thus, (\ref{for-jl2}) implies that
\begin{eqnarray}
\forall p\in C_1, \hspace{0.1in} ||p-f^{-1}(\bar{c}_1)||^2&\leq& \frac{(1+\epsilon)^2}{1-\epsilon} \hspace{0.05in} \lambda^2 (r'_{opt})^2\nonumber\\
&\leq&  \frac{(1+\epsilon)^3}{1-\epsilon}\lambda^2 r^2_{opt},\label{for-jl7}
\end{eqnarray}
where the last inequality comes from (\ref{for-jl6}). The inequality (\ref{for-jl7}) indicates that the ball $\mathbb{B}\big(f^{-1}(\bar{c}_1), \sqrt{\frac{(1+\epsilon)^3}{1-\epsilon}}\lambda r_{opt}\big)$ covers the whole set $C_1$, {\em i.e.,} $f^{-1}(\bar{c}_1)$ yields a $\sqrt{\frac{(1+\epsilon)^3}{1-\epsilon}}\lambda$-approximate solution for the instance $(P, \gamma)$.
 \end{proof}

%\subsection{$k$-Center Clustering with Outliers}
%We extend the idea of Algorithm~\ref{alg-meb} to the case of general $k\geq 1$ (see Algorithm~\ref{alg-kcenter}). 

\textbf{From MEB to $k$-center.} The proof of Theorem~\ref{the-jl} can be  extended for $k$-center clustering. We also use $r_{opt}$ and $r'_{opt}$ to denote the optimal radii of the instances $(P, \gamma)$ and $\big(f(P), \gamma\big)$, respectively. Then, we have the same claim as  (\ref{for-jl6}), $r'_{opt}\leq \sqrt{1+\epsilon}r_{opt}$. 
Consequently, the inequality (\ref{for-jl7}) is replaced by: for any $1\leq j\leq k$ and any $p\in C_j$,
\begin{eqnarray}
  ||p-f^{-1}(\bar{c}_j)||^2&\leq&  \frac{(1+\epsilon)^3}{1-\epsilon}\lambda^2 r^2_{opt}.
\end{eqnarray}
Thus, the set $\cup^k_{j=1}C_j$ is covered by the union of the $k$ balls $\cup^k_{j=1}\mathbb{B}\big(f^{-1}(\bar{c}_j), \sqrt{\frac{(1+\epsilon)^3}{1-\epsilon}}\lambda r_{opt}\big)$, {\em i.e.,} $\{f^{-1}(\bar{c}_1), f^{-1}(\bar{c}_2), \cdots, f^{-1}(\bar{c}_k)\}$ yields a $\sqrt{\frac{(1+\epsilon)^3}{1-\epsilon}}\lambda$-approximate solution for the instance $(P, \gamma)$.
%and $\{C_1, C_2, \cdots, C_k\}$ are the $k$ clusters.

\newcounter{sd4}
\begin{figure*}[]
	\begin{center}
		\centerline{\includegraphics[width=0.5\columnwidth]{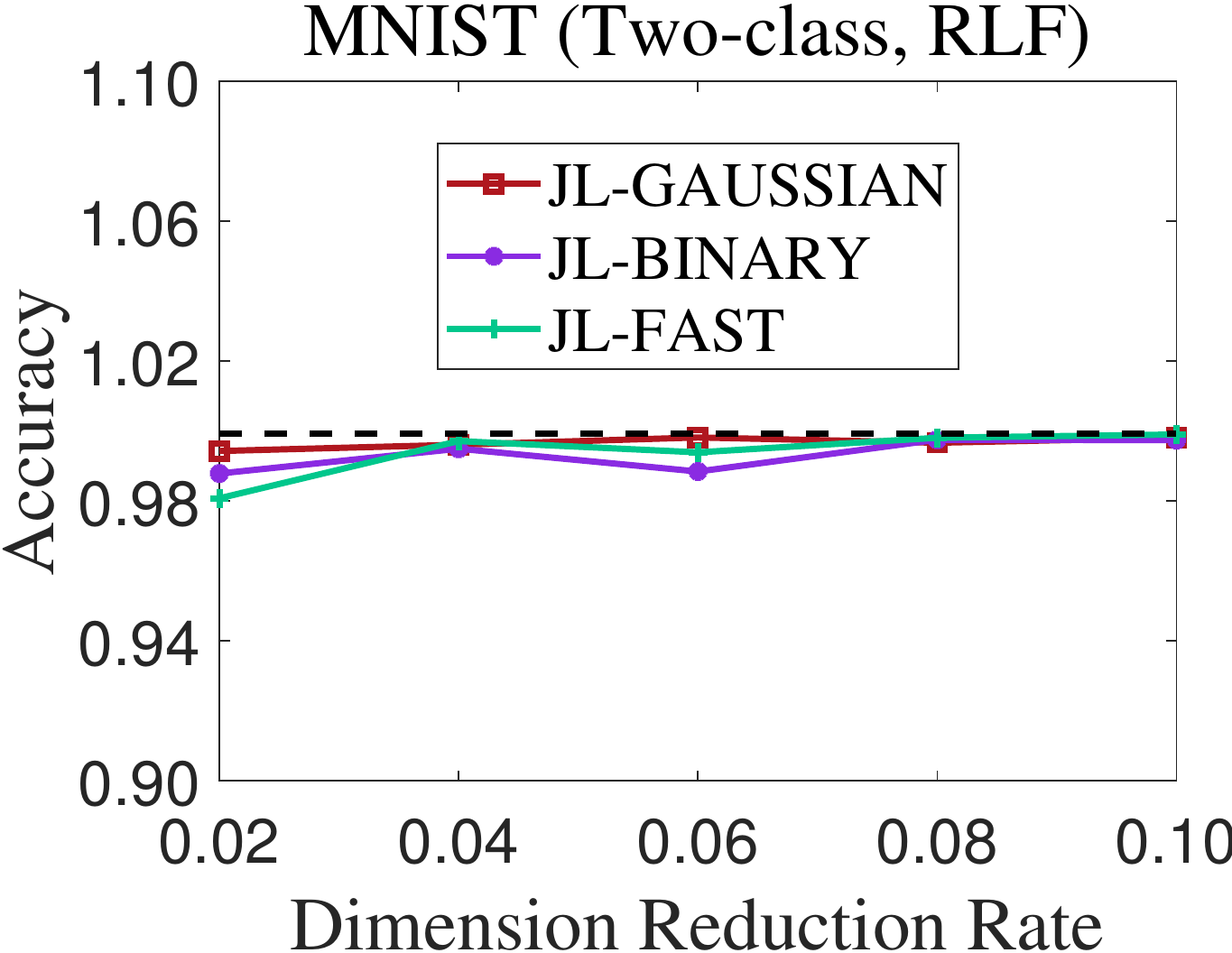}
			\includegraphics[width=0.5\columnwidth]{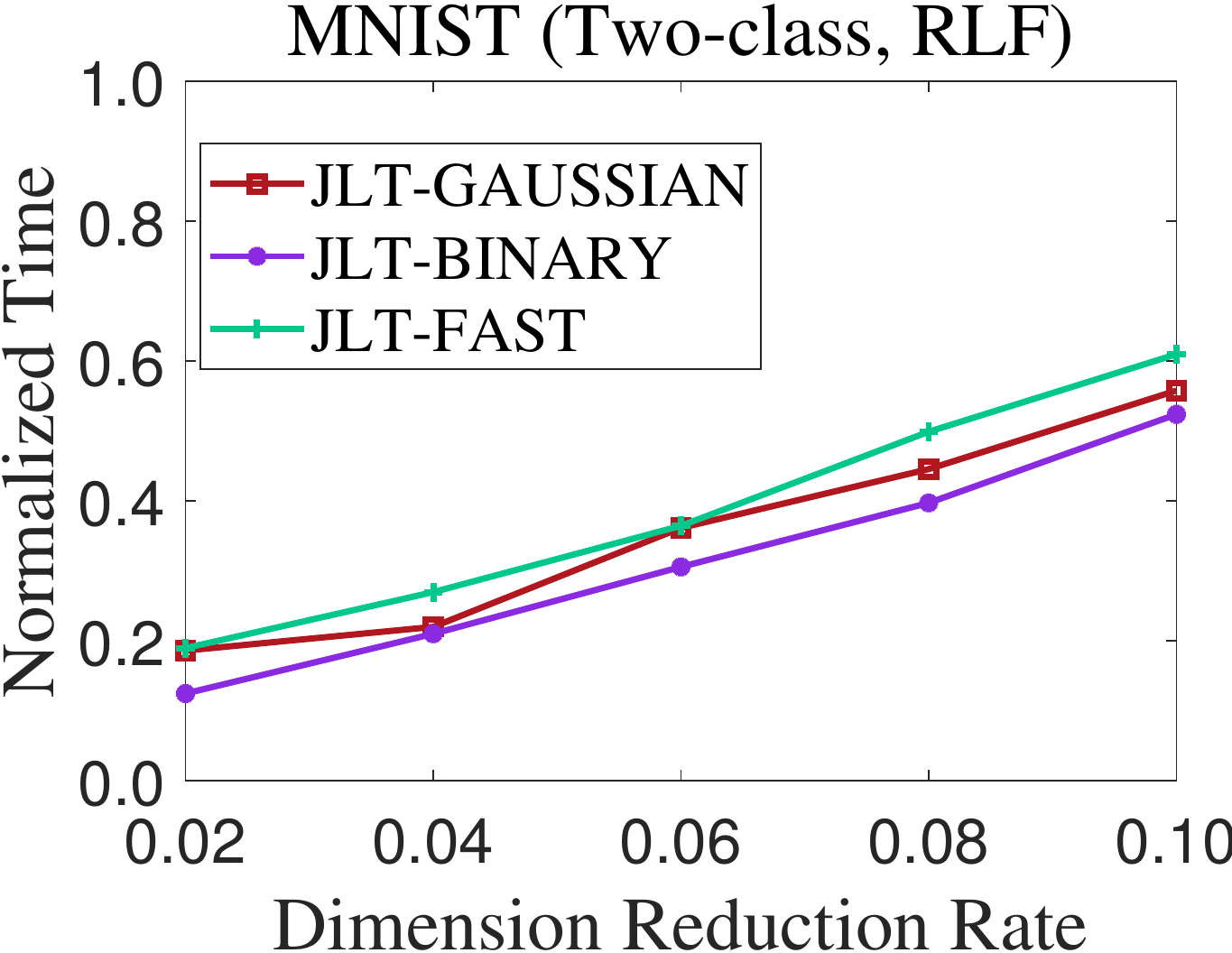}
			\includegraphics[width=0.5\columnwidth]{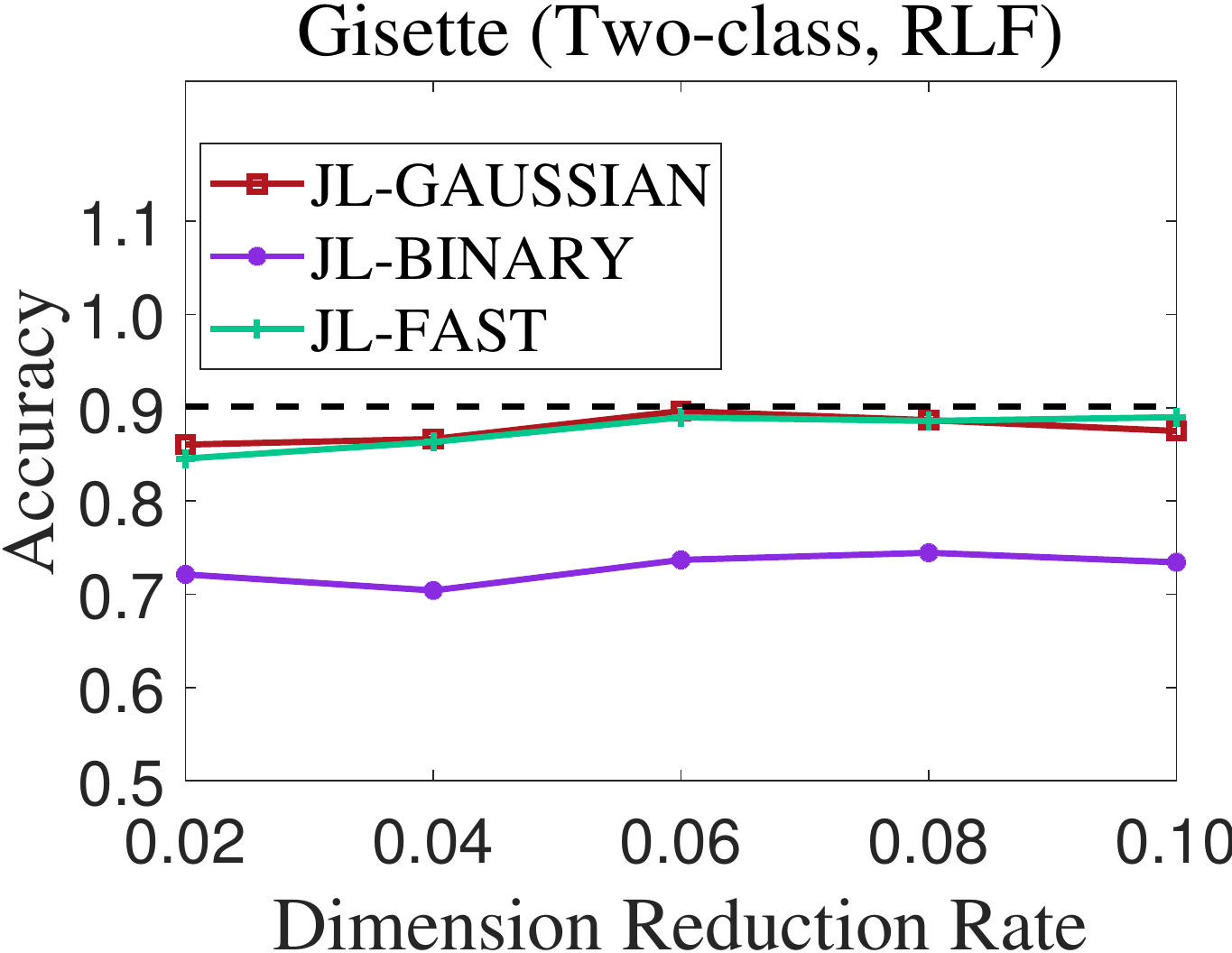}
			\includegraphics[width=0.5\columnwidth]{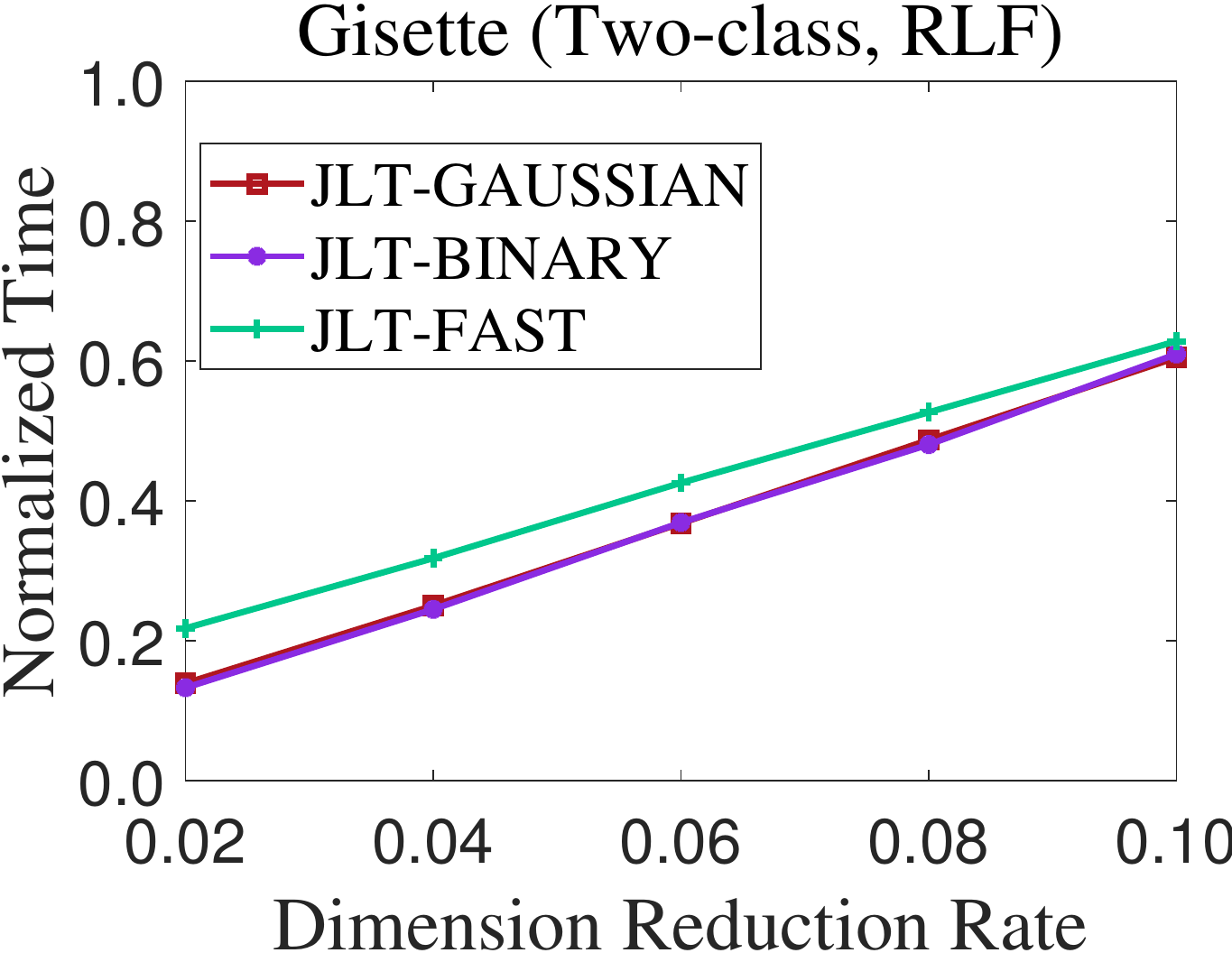}
		}
		\centerline{ \hspace{-0.8in}\hfill \stepcounter{sd4} (\alph{sd4})\hfill \stepcounter{sd4} (\alph{sd4})\hfill \stepcounter{sd4} (\alph{sd4})\hfill \stepcounter{sd4} (\alph{sd4})\hspace{0.8in}}
		%\caption{}
		%\label{icml-historical2}
%	\end{center}
%	\caption{The comparison on different JL transform methods for two-class SVM with outliers. The horizontal lines in (a) and (c) indicate the obtained classification accuracies without dimension reduction.}
%	\label{fig-exp4}
%\end{figure*}
%\newcounter{sd6}
%\begin{figure*}[ht]
%	\begin{center}
		\centerline{
		\includegraphics[width=0.5\columnwidth]{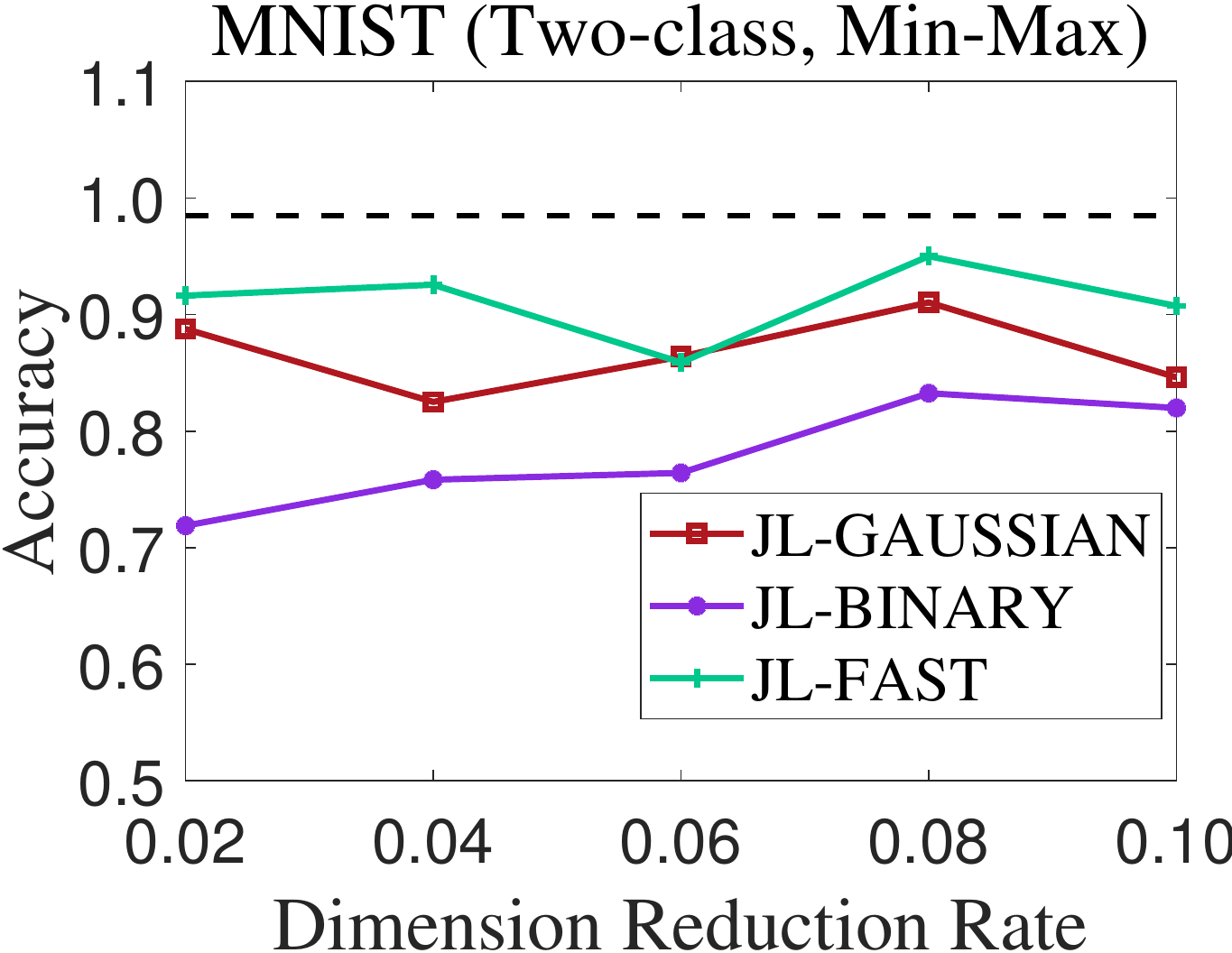}
			\includegraphics[width=0.5\columnwidth]{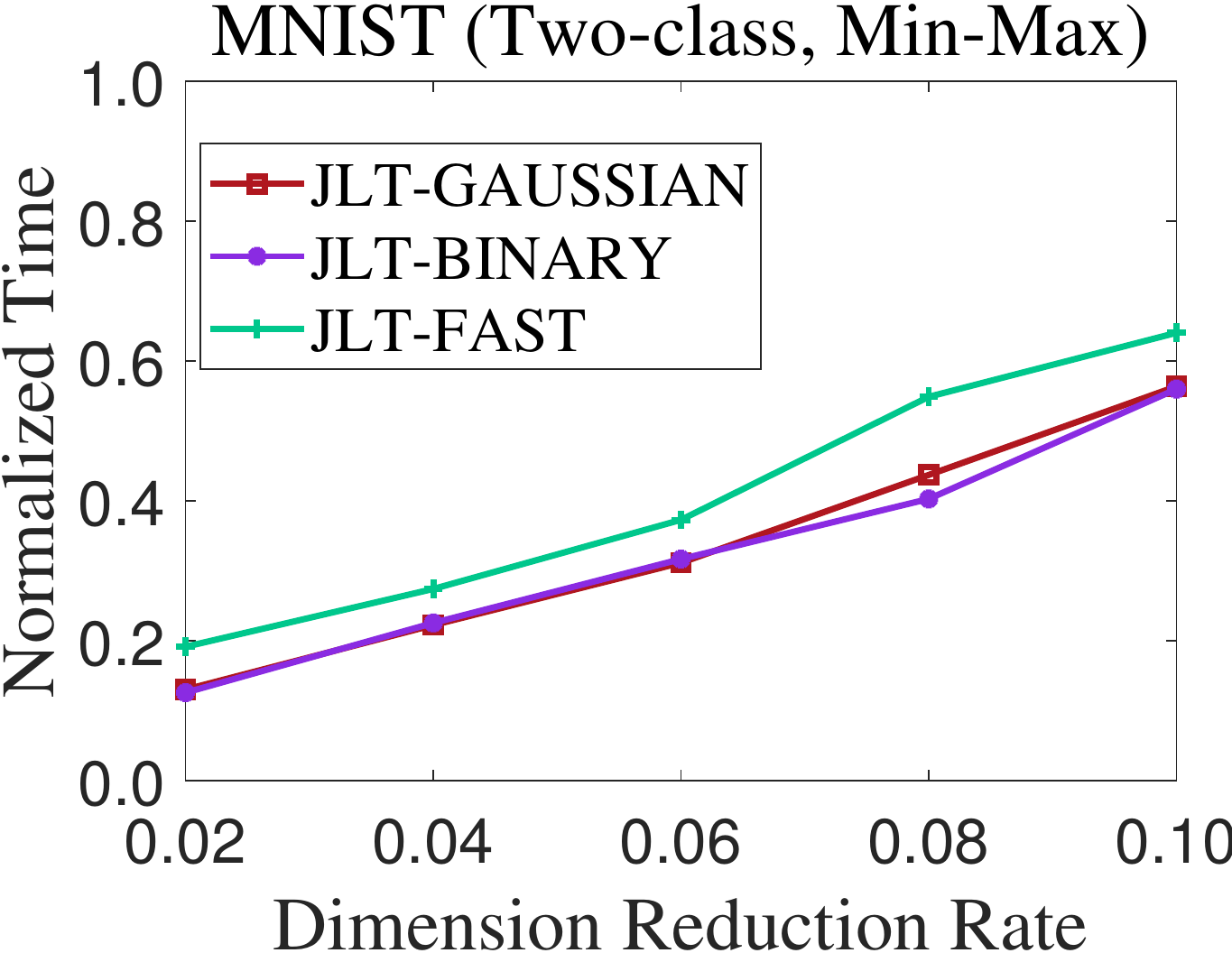}
			\includegraphics[width=0.5\columnwidth]{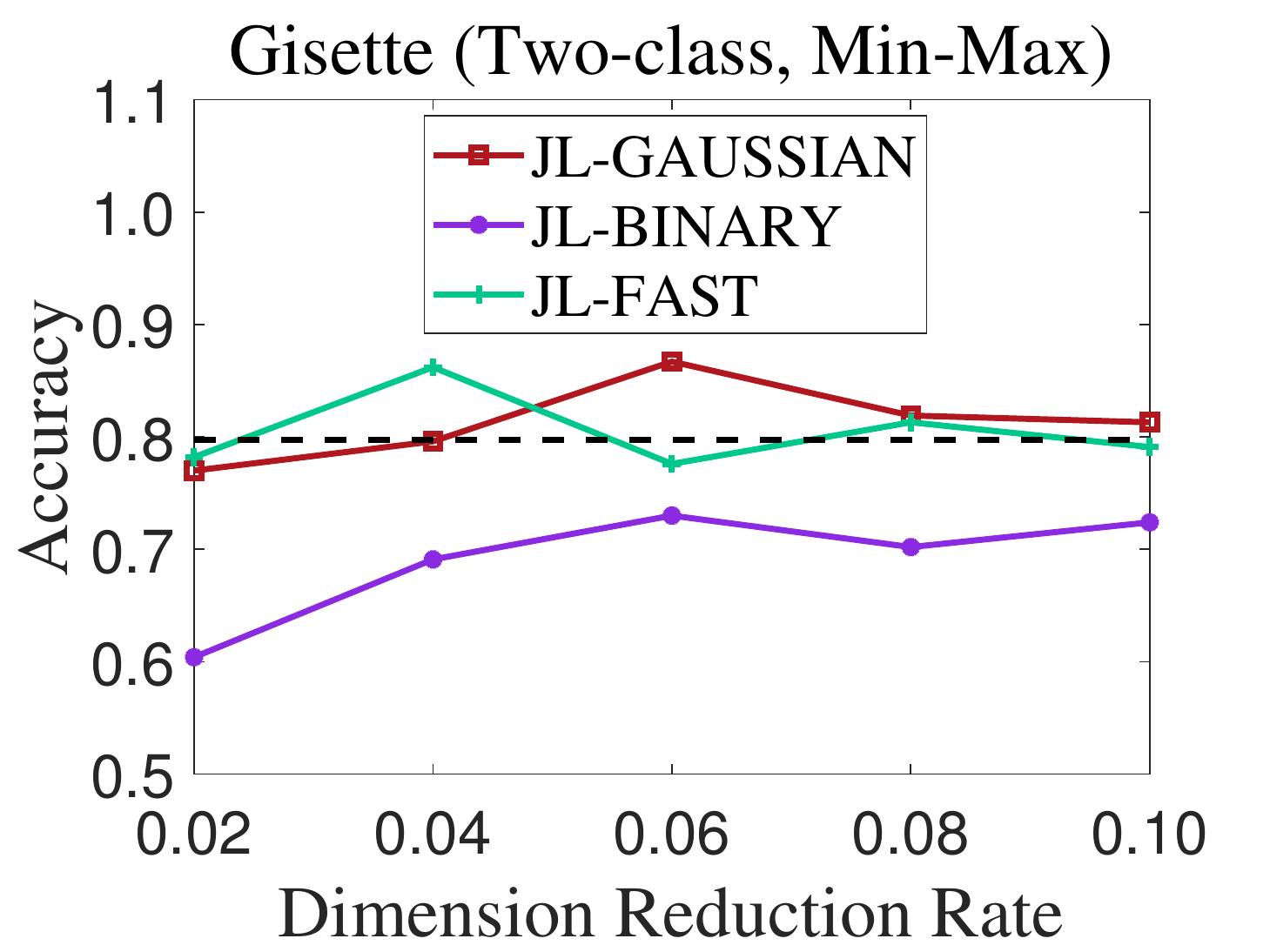}
			\includegraphics[width=0.5\columnwidth]{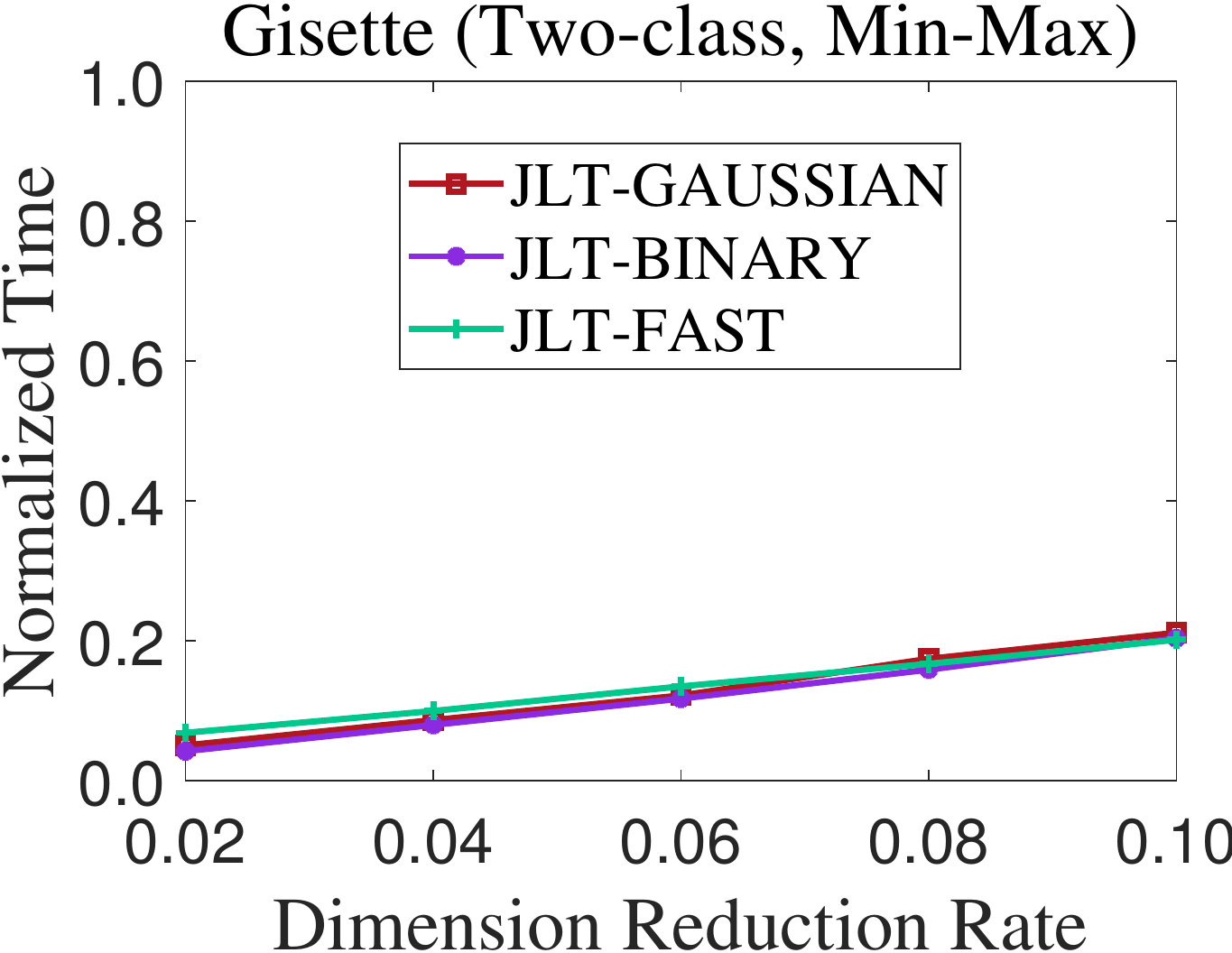}
		}
		\centerline{ \hspace{-0.8in}\hfill \stepcounter{sd4} (\alph{sd4})\hfill \stepcounter{sd4} (\alph{sd4})\hfill \stepcounter{sd4} (\alph{sd4})\hfill \stepcounter{sd4} (\alph{sd4})\hspace{0.8in}}
%		\centerline{ \hfill \stepcounter{sd6} (\alph{sd6})\hfill \stepcounter{sd6} (\alph{sd6})\hfill \stepcounter{sd6} }
		%\caption{}
		%\label{icml-historical2}
	\end{center}
	\vspace{-0.3in}
	\caption{The comparison on different JL transform methods for two-class SVM under \textsc{RLF} attack and \textsc{MIN-MAX} attack (averaged across $15$ trials). The horizontal lines  indicates the obtained classification accuracy without dimension reduction. For each instance, the runtimes are normalized over the time of directly running $\mathcal{A}$ without dimension reduction. }
	\label{fig-exp6}
		\vspace{-0.2in}
\end{figure*}

%\begin{theorem}\textbf{(radius preservation for $k$-center clustering with outliers)}
%\label{the-kcenter}
%Let $\lambda\geq 1$. Suppose the algorithm $\mathcal{A}$ used in Step 1 of Algorithm~\ref{alg-kcenter} yields a $\lambda$-approximate solution of $k$-center clustering with outliers. The returned set of $k$ points $\{f^{-1}(\bar{c}_1), f^{-1}(\bar{c}_2), \cdots, f^{-1}(\bar{c}_k)\}$ of Algorithm~\ref{alg-kcenter} yields a $\lambda\sqrt{\frac{(1+\epsilon)^3}{1-\epsilon}}$-approximate solution for the original instance $(P,\gamma)$ in $\mathbb{R}^d$ with constant probability.
%\end{theorem}
%\begin{remark}
%\label{rem-kcenter}
%If allowing a relaxation on the number of outliers (say slightly larger than $\gamma n$), \cite{huang2018epsilon,DBLP:conf/esa/DingYW19} showed that one can run existing algorithm on a random sample of size (roughly) $O(\frac{kd\log k}{\gamma})$ instead of the original input data so as to achieve an approximate solution. From Theorem~\ref{the-kcenter}, we immediately see that the sample size can be reduced after applying JL transform (the dimensionality $d$ is reduced to $O(\frac{\log n}{\epsilon^2})$ in the sample size). 
%\end{remark}

\textbf{Time complexity.} The time complexity of Algorithm~\ref{alg-kcenter} also consists of the three parts as  (\ref{for-time-one-class}). We reduce the dimensionality $d$ to be $\tilde{d}=O(\frac{\log n}{\epsilon^2})$. If the black box algorithm $\mathcal{A}$ of Algorithm~\ref{alg-kcenter} has the time complexity $\Gamma(n, d)$ for an instance of $n$ points in $d$-dimensional space, the total complexity of Algorithm~\ref{alg-kcenter} should be
\begin{eqnarray}
\mathtt{Time}_{JL}+\Gamma(n, \tilde{d})+\mathtt{Time}_{rec},\label{for-time-kcenter}
\end{eqnarray}
where $\mathtt{Time}_{JL}=O(\frac{1}{\epsilon^2}nd\log n)$ and $\mathtt{Time}_{rec}=O\Big(\frac{1}{\epsilon}\cdot (n\cdot \tilde{d}+k\cdot d)\Big)$ (because BC's algorithm runs in $O(\frac{1}{\epsilon})$ steps in the $\tilde{d}$-dimensional space). Similar with  (\ref{for-time-one-class}), the second term $\Gamma(n, \tilde{d})$ often dominates the whole complexity in practice.

\vspace{-0.15in}
\section{Experiments}
\vspace{-0.1in}

All the experimental results were obtained on a Windows workstation with 2.8GHz Intel  Core  i5-840 and 8GB main memory; the algorithms were implemented in Matlab R2018b.  We compare the three representative JL transform methods from \citep{dasgupta2003elementary,achlioptas2003database,ailon2009fast} that are mentioned in Section~\ref{sec-intro}; we name them as \textsc{JLT-Gaussian}, \textsc{JLT-binary}, and \textsc{JLT-fast}, respectively. 

We consider two-class SVM first. We use the popular  implementation of linear SVM   from \citep{journals/tist/ChangL11} as the black box algorithm $\mathcal{A}$. We run the experiments on two widely used high dimensional real  datasets.   \textbf{Gisette} \citep{guyon2005result} contains two classes of  $n=13,500$ vectors in $\mathbb{R}^{5000}$.  \textbf{MNIST}~\citep{lecun98} contains $n=60,000$ handwritten digit images from $0$ to $9$, where each image is represented by a $784$-dimensional vector (so any pair of two digits can form an instance of two-class SVM). 
%Each class of \textbf{MNIST} and \textbf{GISETTE} forms an individual instance.  
Each instance is randomly partitioned into two equal-sized parts respectively for training and testing. We apply two methods to generate the outliers. \textbf{(1)} We randomly select $10\%$ pairs from the data and flip their labels (termed \textbf{\textsc{RLF}}). \textbf{(2)} Further, we use the publicly available adversarial attack software~\citep{DBLP:journals/corr/abs-1811-00741} (termed \textbf{\textsc{MIN-MAX}}) to generate $10\%$ adversarial outliers for each instance.  
%datasets and the \textbf{GISETTE} datasets, 
%and the results are shown in the second line of Figure~\ref{fig-exp6}. 
 
 We vary the dimension reduction rate from $2\%$ to $10\%$ ({\em e.g.,} if the  rate is $2\%$, that means we reduce the dimension from $d$ to $\tilde{d}=2\%*d$).  
Their results are shown in Figure~\ref{fig-exp6}. 
We can see the three methods can achieve roughly the same runtimes that are significantly lower than those without dimension reduction (we count the runtime for all the three parts as (\ref{for-time-one-class})). And meanwhile, for most cases  \textsc{JLT-Gaussian} and \textsc{JLT-fast} are very effective and can achieve close accuracies to the ones without dimension reduction, even when the dimension reduction rate is only $2\%$. 
%We can see that \textsc{JLT-Gaussian} outperforms the other two methods in general; \textsc{JLT-binary} and \textsc{JLT-fast} either take the largest running time or achieve the worst classification accuracy. 
An interesting observation is that for the higer dimensional dataset \textbf{Gisette}, the accuracies of some instances are improved after dimension reduction. We believe that one possible reason is that dimension reduction can slightly relieve the data redundancy issue which may potentially mislead classification.

%it is due to the high redundancy of the information in high dimensional space, which could potentially destroy the classification results; the experimental results reveal that dimension reduction can relieve this issue for some cases. 

We also try to run other dimension reduction methods 
such as PCA, approximate PCA~\citep{boutsidis2014randomized}, and the importance sampling based feature selection~\citep{paul2015feature}. However,  the PCA and approximate PCA methods both take too long runtimes that their normalized running time ratios are even larger than $1$; that means we do not gain any benefit from the dimension reduction. In addition, the feature selection method often yields poor classification results ({\em e.g.,} around $70\%$ accuracy for \textbf{MNIST} under \textsc{RLF} attack), and we believe the reason is that importance sampling may incur high distortion  when outliers exist.

Due to the space limit, we leave more experimental results to our appendix.

\vspace{-0.12in}

\section{Future Work}
\vspace{-0.12in}

%\vspace{-0.03in}
%In this paper, we study the effectiveness of JL transform for dealing with two fundamental high-dimensional optimization problems, SVM and $k$-center clustering, in the presence of outliers.  
In future, we can consider applying random projection to other robust optimization problems with (potentially adversarial) outliers. 
Also, it is interesting to study the theoretical quality guarantees of the (approximate) PCA and feature selection methods when input dataset contains outliers.

 \newpage

\bibliography{random}

\begin{thebibliography}{53}
\providecommand{\natexlab}[1]{#1}
\providecommand{\url}[1]{\texttt{#1}}
\expandafter\ifx\csname urlstyle\endcsname\relax
  \providecommand{\doi}[1]{doi: #1}\else
  \providecommand{\doi}{doi: \begingroup \urlstyle{rm}\Url}\fi

\bibitem[Achlioptas(2003)]{achlioptas2003database}
Dimitris Achlioptas.
\newblock Database-friendly random projections: {J}ohnson-{L}indenstrauss with
  binary coins.
\newblock \emph{Journal of computer and System Sciences}, 66\penalty0
  (4):\penalty0 671--687, 2003.

\bibitem[Agarwal et~al.(2007)Agarwal, Har{-}Peled, and
  Yu]{DBLP:conf/compgeom/AgarwalHY07}
Pankaj~K. Agarwal, Sariel Har{-}Peled, and Hai Yu.
\newblock Embeddings of surfaces, curves, and moving points in euclidean space.
\newblock In \emph{Proceedings of the 23rd {ACM} Symposium on Computational
  Geometry}, pages 381--389, 2007.

\bibitem[Ailon and Chazelle(2009)]{ailon2009fast}
Nir Ailon and Bernard Chazelle.
\newblock The fast {J}ohnson--{L}indenstrauss transform and approximate nearest
  neighbors.
\newblock \emph{SIAM Journal on computing}, 39\penalty0 (1):\penalty0 302--322,
  2009.

\bibitem[Arriaga and Vempala(2006)]{DBLP:journals/ml/ArriagaV06}
Rosa~I. Arriaga and Santosh~S. Vempala.
\newblock An algorithmic theory of learning: Robust concepts and random
  projection.
\newblock \emph{Machine Learning}, 63\penalty0 (2):\penalty0 161--182, 2006.

\bibitem[Balcan et~al.(2006)Balcan, Blum, and Vempala]{balcan2006kernels}
Maria-Florina Balcan, Avrim Blum, and Santosh Vempala.
\newblock Kernels as features: On kernels, margins, and low-dimensional
  mappings.
\newblock \emph{Machine Learning}, 65\penalty0 (1):\penalty0 79--94, 2006.

\bibitem[Biggio and Roli(2018)]{biggio2018wild}
Battista Biggio and Fabio Roli.
\newblock Wild patterns: Ten years after the rise of adversarial machine
  learning.
\newblock \emph{Pattern Recognition}, 84:\penalty0 317--331, 2018.

\bibitem[Boutsidis et~al.(2014)Boutsidis, Zouzias, Mahoney, and
  Drineas]{boutsidis2014randomized}
Christos Boutsidis, Anastasios Zouzias, Michael~W Mahoney, and Petros Drineas.
\newblock Randomized dimensionality reduction for $ k $-means clustering.
\newblock \emph{IEEE Transactions on Information Theory}, 61\penalty0
  (2):\penalty0 1045--1062, 2014.

\bibitem[B\u{a}doiu and Clarkson(2003)]{badoiu2003smaller}
Mihai B\u{a}doiu and Kenneth~L. Clarkson.
\newblock Smaller core-sets for balls.
\newblock In \emph{SODA}, pages 801--802, 2003.

\bibitem[B\u{a}doiu et~al.(2002)B\u{a}doiu, Har-Peled, and Indyk]{BHI}
Mihai B\u{a}doiu, Sariel Har-Peled, and Piotr Indyk.
\newblock Approximate clustering via core-sets.
\newblock In \emph{STOC}, pages 250--257, 2002.

\bibitem[Ceccarello et~al.(2019)Ceccarello, Pietracaprina, and
  Pucci]{DBLP:journals/pvldb/CeccarelloPP19}
Matteo Ceccarello, Andrea Pietracaprina, and Geppino Pucci.
\newblock Solving k-center clustering (with outliers) in mapreduce and
  streaming, almost as accurately as sequentially.
\newblock \emph{Proc. {VLDB} Endow.}, 12\penalty0 (7):\penalty0 766--778, 2019.

\bibitem[Chakrabarty et~al.(2016)Chakrabarty, Goyal, and
  Krishnaswamy]{DBLP:conf/icalp/ChakrabartyGK16}
Deeparnab Chakrabarty, Prachi Goyal, and Ravishankar Krishnaswamy.
\newblock The non-uniform k-center problem.
\newblock In \emph{ICALP}, pages 67:1--67:15, 2016.

\bibitem[Chang and Lin(2011)]{journals/tist/ChangL11}
Chih-Chung Chang and Chih-Jen Lin.
\newblock {LIBSVM}: {A} library for support vector machines.
\newblock \emph{ACM TIST}, 2\penalty0 (3), 2011.

\bibitem[Charikar et~al.(2001)Charikar, Khuller, Mount, and
  Narasimhan]{charikar2001algorithms}
Moses Charikar, Samir Khuller, David~M Mount, and Giri Narasimhan.
\newblock Algorithms for facility location problems with outliers.
\newblock In \emph{SODA}, pages 642--651, 2001.

\bibitem[Clarkson(2010)]{C10}
Kenneth~L. Clarkson.
\newblock Coresets, sparse greedy approximation, and the {F}rank-{W}olfe
  algorithm.
\newblock \emph{ACM Transactions on Algorithms}, 6\penalty0 (4):\penalty0 63,
  2010.

\bibitem[Cohen et~al.(2015)Cohen, Elder, Musco, Musco, and
  Persu]{cohen2015dimensionality}
Michael~B Cohen, Sam Elder, Cameron Musco, Christopher Musco, and Madalina
  Persu.
\newblock Dimensionality reduction for k-means clustering and low rank
  approximation.
\newblock In \emph{Proceedings of the forty-seventh annual ACM symposium on
  Theory of computing}, pages 163--172, 2015.

\bibitem[Cortes and Vapnik(1995)]{mach:Cortes+Vapnik:1995}
Corinna Cortes and Vladimir Vapnik.
\newblock Support-vector networks.
\newblock \emph{Machine Learning}, 20:\penalty0 273, 1995.

\bibitem[Crisp and Burges(1999)]{conf/nips/CrispB99}
David~J. Crisp and Christopher J.~C. Burges.
\newblock A geometric interpretation of v-{SVM} classifiers.
\newblock In \emph{NIPS}. The MIT Press, 1999.

\bibitem[Cunningham and Ghahramani(2015)]{cunningham2015linear}
John~P Cunningham and Zoubin Ghahramani.
\newblock Linear dimensionality reduction: Survey, insights, and
  generalizations.
\newblock \emph{The Journal of Machine Learning Research}, 16\penalty0
  (1):\penalty0 2859--2900, 2015.

\bibitem[Dasgupta et~al.(2010)Dasgupta, Kumar, and
  Sarl{\'{o}}s]{DBLP:conf/stoc/DasguptaKS10}
Anirban Dasgupta, Ravi Kumar, and Tam{\'{a}}s Sarl{\'{o}}s.
\newblock A sparse {J}ohnson: {L}indenstrauss transform.
\newblock In \emph{STOC}, pages 341--350, 2010.

\bibitem[Dasgupta and Gupta(2003)]{dasgupta2003elementary}
Sanjoy Dasgupta and Anupam Gupta.
\newblock An elementary proof of a theorem of {J}ohnson and {L}indenstrauss.
\newblock \emph{Random Structures \& Algorithms}, 22\penalty0 (1):\penalty0
  60--65, 2003.

\bibitem[Ding and Xu(2015)]{ding2015random}
Hu~Ding and Jinhui Xu.
\newblock Random gradient descent tree: A combinatorial approach for svm with
  outliers.
\newblock In \emph{AAAI}, pages 2561--2567, 2015.

\bibitem[Ding et~al.(2019)Ding, Yu, and Wang]{DBLP:conf/esa/DingYW19}
Hu~Ding, Haikuo Yu, and Zixiu Wang.
\newblock Greedy strategy works for k-center clustering with outliers and
  coreset construction.
\newblock In \emph{ESA}, pages 40:1--40:16, 2019.

\bibitem[Durrant and Kab{\'{a}}n(2012)]{jltutorial}
Robert~J. Durrant and Ata Kab{\'{a}}n.
\newblock Random projections for machine learning and data mining: Theory \&
  applications.
\newblock \emph{ECML PKDD 2012 tutorial}, 2012.

\bibitem[Erfani et~al.(2016)Erfani, Rajasegarar, Karunasekera, and
  Leckie]{DBLP:journals/pr/ErfaniRKL16}
Sarah~M. Erfani, Sutharshan Rajasegarar, Shanika Karunasekera, and Christopher
  Leckie.
\newblock High-dimensional and large-scale anomaly detection using a linear
  one-class {SVM} with deep learning.
\newblock \emph{Pattern Recognition}, 58:\penalty0 121--134, 2016.

\bibitem[G{\"a}rtner and Jaggi(2009)]{GJ09}
Bernd G{\"a}rtner and Martin Jaggi.
\newblock Coresets for polytope distance.
\newblock In \emph{SoCG}, pages 33--42, 2009.

\bibitem[Gilbert(1966)]{gilbert1966iterative}
Elmer~G. Gilbert.
\newblock An iterative procedure for computing the minimum of a quadratic form
  on a convex set.
\newblock \emph{SIAM Journal on Control}, 4\penalty0 (1):\penalty0 61--80,
  1966.

\bibitem[Guyon et~al.(2005)Guyon, Gunn, Ben-Hur, and Dror]{guyon2005result}
Isabelle Guyon, Steve Gunn, Asa Ben-Hur, and Gideon Dror.
\newblock Result analysis of the nips 2003 feature selection challenge.
\newblock In \emph{Advances in neural information processing systems}, pages
  545--552, 2005.

\bibitem[Johnson and Lindenstrauss(1984)]{johnson1984extensions}
William~B Johnson and Joram Lindenstrauss.
\newblock Extensions of lipschitz mappings into a hilbert space.
\newblock \emph{Contemporary mathematics}, 26\penalty0 (189-206):\penalty0 1,
  1984.

\bibitem[Kane and Nelson(2014)]{DBLP:journals/jacm/KaneN14}
Daniel~M. Kane and Jelani Nelson.
\newblock Sparser {J}ohnson-{L}indenstrauss transforms.
\newblock \emph{J. {ACM}}, 61\penalty0 (1):\penalty0 4:1--4:23, 2014.

\bibitem[Kerber and Raghvendra(2014)]{kerber2014approximation}
Michael Kerber and Sharath Raghvendra.
\newblock Approximation and streaming algorithms for projective clustering via
  random projections.
\newblock \emph{arXiv preprint arXiv:1407.2063}, 2014.

\bibitem[Koh et~al.(2018)Koh, Steinhardt, and
  Liang]{DBLP:journals/corr/abs-1811-00741}
Pang~Wei Koh, Jacob Steinhardt, and Percy Liang.
\newblock Stronger data poisoning attacks break data sanitization defenses.
\newblock \emph{CoRR}, abs/1811.00741, 2018.

\bibitem[Krizhevsky(2009)]{Cifar10}
Alex Krizhevsky.
\newblock Learning multiple layers of features from tiny images.
\newblock 2009.

\bibitem[Kumar et~al.(2008)Kumar, Bhattacharya, and
  Hariharan]{kumar2008randomized}
Krishnan Kumar, Chiru Bhattacharya, and Ramesh Hariharan.
\newblock A randomized algorithm for large scale support vector learning.
\newblock In \emph{Advances in Neural Information Processing Systems}, pages
  793--800, 2008.

\bibitem[Makarychev et~al.(2019)Makarychev, Makarychev, and
  Razenshteyn]{makarychev2019performance}
Konstantin Makarychev, Yury Makarychev, and Ilya Razenshteyn.
\newblock Performance of {J}ohnson-{L}indenstrauss transform for k-means and
  k-medians clustering.
\newblock In \emph{Proceedings of the 51st Annual ACM SIGACT Symposium on
  Theory of Computing}, pages 1027--1038, 2019.

\bibitem[Malkomes et~al.(2015)Malkomes, Kusner, Chen, Weinberger, and
  Moseley]{malkomes2015fast}
Gustavo Malkomes, Matt~J Kusner, Wenlin Chen, Kilian~Q Weinberger, and Benjamin
  Moseley.
\newblock Fast distributed k-center clustering with outliers on massive data.
\newblock In \emph{Advances in Neural Information Processing Systems}, pages
  1063--1071, 2015.

\bibitem[McCutchen and Khuller(2008)]{mccutchen2008streaming}
Richard~Matthew McCutchen and Samir Khuller.
\newblock Streaming algorithms for k-center clustering with outliers and with
  anonymity.
\newblock In \emph{Approximation, Randomization and Combinatorial Optimization.
  Algorithms and Techniques}, pages 165--178. Springer, 2008.

\bibitem[Paul et~al.(2014)Paul, Boutsidis, Magdon{-}Ismail, and
  Drineas]{DBLP:journals/tkdd/PaulBMD14}
Saurabh Paul, Christos Boutsidis, Malik Magdon{-}Ismail, and Petros Drineas.
\newblock Random projections for linear support vector machines.
\newblock \emph{{TKDD}}, 8\penalty0 (4):\penalty0 22:1--22:25, 2014.

\bibitem[Paul et~al.(2015)Paul, Magdon-Ismail, and Drineas]{paul2015feature}
Saurabh Paul, Malik Magdon-Ismail, and Petros Drineas.
\newblock Feature selection for linear svm with provable guarantees.
\newblock In \emph{Artificial Intelligence and Statistics}, pages 735--743,
  2015.

\bibitem[Platt(1999)]{platt99}
J.~Platt.
\newblock Fast training of support vector machines using sequential minimal
  optimization.
\newblock In \emph{Advances in Kernel Methods --- Support Vector Learning},
  pages 185--208. {MIT} Press, 1999.

\bibitem[Rahimi and Recht(2008)]{rahimi2008random}
Ali Rahimi and Benjamin Recht.
\newblock Random features for large-scale kernel machines.
\newblock In \emph{Advances in neural information processing systems}, pages
  1177--1184, 2008.

\bibitem[Rousseeuw and Leroy(1987)]{books/wi/RousseeuwL87}
Peter~J. Rousseeuw and Annick Leroy.
\newblock \emph{Robust Regression and Outlier Detection}.
\newblock Wiley, 1987.

\bibitem[Scholkopf et~al.(2000)Scholkopf, Smola, Muller, and Bartlett]{bb57389}
B.~Scholkopf, A.~J. Smola, K.~R. Muller, and P.~L. Bartlett.
\newblock New support vector algorithms.
\newblock \emph{Neural Computation}, 12:\penalty0 1207--1245, 2000.

\bibitem[Sch{\"{o}}lkopf et~al.(1999)Sch{\"{o}}lkopf, Williamson, Smola,
  Shawe{-}Taylor, and Platt]{DBLP:conf/nips/ScholkopfWSSP99}
Bernhard Sch{\"{o}}lkopf, Robert~C. Williamson, Alexander~J. Smola, John
  Shawe{-}Taylor, and John~C. Platt.
\newblock Support vector method for novelty detection.
\newblock In \emph{Advances in Neural Information Processing Systems}, pages
  582--588, 1999.

\bibitem[Sheehy(2014)]{DBLP:conf/compgeom/Sheehy14}
Donald~R. Sheehy.
\newblock The persistent homology of distance functions under random
  projection.
\newblock In \emph{SOCG}, page 328, 2014.

\bibitem[Shi et~al.(2012)Shi, Shen, Hill, and van~den
  Hengel]{DBLP:conf/icml/ShiSHH12}
Qinfeng Shi, Chunhua Shen, Rhys Hill, and Anton van~den Hengel.
\newblock Is margin preserved after random projection?
\newblock In \emph{Proceedings of the 29th International Conference on Machine
  Learning, {ICML}}, 2012.

\bibitem[Suzumura et~al.(2014)Suzumura, Ogawa, Sugiyama, and
  Takeuchi]{icml2014c2_suzumura14}
Shinya Suzumura, Kohei Ogawa, Masashi Sugiyama, and Ichiro Takeuchi.
\newblock Outlier path: A homotopy algorithm for robust svm.
\newblock In \emph{ICML}, pages 1098--1106. JMLR Workshop and Conference
  Proceedings, 2014.

\bibitem[Tan et~al.(2006)Tan, Steinbach, and Kumar]{tan2006introduction}
Pang-Ning Tan, Michael Steinbach, and Vipin Kumar.
\newblock \emph{Introduction to Data Mining}.
\newblock 2006.

\bibitem[Xu et~al.(2017)Xu, Cao, Hu, and
  Pr{\'{\i}}ncipe]{DBLP:journals/pr/XuCHP17}
Guibiao Xu, Zheng Cao, Bao{-}Gang Hu, and Jos{\'{e}}~C. Pr{\'{\i}}ncipe.
\newblock Robust support vector machines based on the rescaled hinge loss
  function.
\newblock \emph{Pattern Recognition}, 63:\penalty0 139--148, 2017.

\bibitem[Xu et~al.(2006)Xu, Crammer, and Schuurmans]{conf/aaai/XuCS06}
Linli Xu, Koby Crammer, and Dale Schuurmans.
\newblock Robust support vector machine training via convex outlier ablation.
\newblock AAAI Press, 2006.

\bibitem[Yang et~al.(2010)Yang, Xu, White, Schuurmans, and Yu]{yang2010relaxed}
Min Yang, Linli Xu, Martha White, Dale Schuurmans, and Yao-liang Yu.
\newblock Relaxed clipping: A global training method for robust regression and
  classification.
\newblock In \emph{Advances in neural information processing systems}, pages
  2532--2540, 2010.

\bibitem[Y.LeCun and P.Haffner(1998)]{lecun98}
Y.Bengio Y.LeCun, L.Bottou and P.Haffner.
\newblock Gradient-based learning applied to document recognition.
\newblock \emph{Proceedings of the IEEE, 86(11):2278-2324, November}, 1998.

\bibitem[Zhang et~al.(2013)Zhang, Mahdavi, Jin, Yang, and
  Zhu]{zhang2013recovering}
Lijun Zhang, Mehrdad Mahdavi, Rong Jin, Tianbao Yang, and Shenghuo Zhu.
\newblock Recovering the optimal solution by dual random projection.
\newblock In \emph{Conference on Learning Theory}, pages 135--157, 2013.

\bibitem[Zimek et~al.(2012)Zimek, Schubert, and Kriegel]{zimek2012survey}
Arthur Zimek, Erich Schubert, and Hans-Peter Kriegel.
\newblock A survey on unsupervised outlier detection in high-dimensional
  numerical data.
\newblock \emph{Statistical Analysis and Data Mining: The ASA Data Science
  Journal}, 5\penalty0 (5):\penalty0 363--387, 2012.

\end{thebibliography}
\newpage
\appendix

\section{More Experimental Results}
%We generate the synthetic dataset in $\mathbb{R}^{4000}$ with $n = {10}^4$, $z = 10\%n$, and  $k = 20$. We randomly pick $k$ points as the cluster centers inside a hypercube of side length $500$; around each center, we generate a cluster following a Gaussian distribution with standard deviation $20$; we keep the total number of points to be $n-z$ and uniformly generate $z$ outliers at random outside the MEBs of these $k$ clusters. We keep the dimension reduction rate to be $\{6\%, 8\%, 10\%, 12\%, 14\%\}$ and repeat the experiment 15 times for each instance; we show the average normalized running time and radius (over the result without dimension reduction) in Figure~\ref{fig-exp7}. 
%
%We also test the algorithms on two real-world datasets. 
%The \textbf{CIFAR-10}~\cite{Cifar10} dataset consists of $n=60,000$ color images of $k=10$ classes ({\em e.g.}, airplane, bird), with each class having $6000$ images and each image being represented by a $3072$-dimensional vector. 
%The \textbf{MNIST} dataset contains $n=60,000$ handwritten digit images from $0$ to $9$ ({\em i.e.,} $k=10$).  
%For each dataset, we randomly add $10\%n$ outliers outside the MEBs of the clusters. 
%We show the resulting normalized radius and running time in Figure~\ref{fig-exp8} and the average standard deviation of normalized radius in Table~\ref{tab-1}. Similar to the results on the synthetic datasets, \textsc{JLT} and \textsc{PCA} achieve the performances much more stable than the other two methods, and \textsc{JLT} is the fastest one. 

\textbf{One-class SVM with outliers.} We use the one-class SVM algorithm and its implementation from \citep{journals/tist/ChangL11} as the algorithm $\mathcal{A}$. As same as the two-class case, we also use the real datasets \textsc{Gisette} and \textsc{MNIST}. 
To generate the outliers, we first compute the hyperplane $\mathcal{H}$ by running the algorithm $\mathcal{A}$ on the training data, and randomly add $10\%$ outliers on the other side of $\mathcal{H}$.  The results of the three JL transform methods are shown in Figure~\ref{fig-exp2}. Similar with the experimental results for two-class SVM, \textsc{JLT-Gaussian} and \textsc{JLT-fast} are very effective and can achieve comparable (or even higher) accuracies with the ones without dimension reduction.

%and Table~\ref{tab-2}.  
\newcounter{sd2}
\begin{figure*}[ht]
	\begin{center}
		\centerline{\includegraphics[width=0.5\columnwidth]{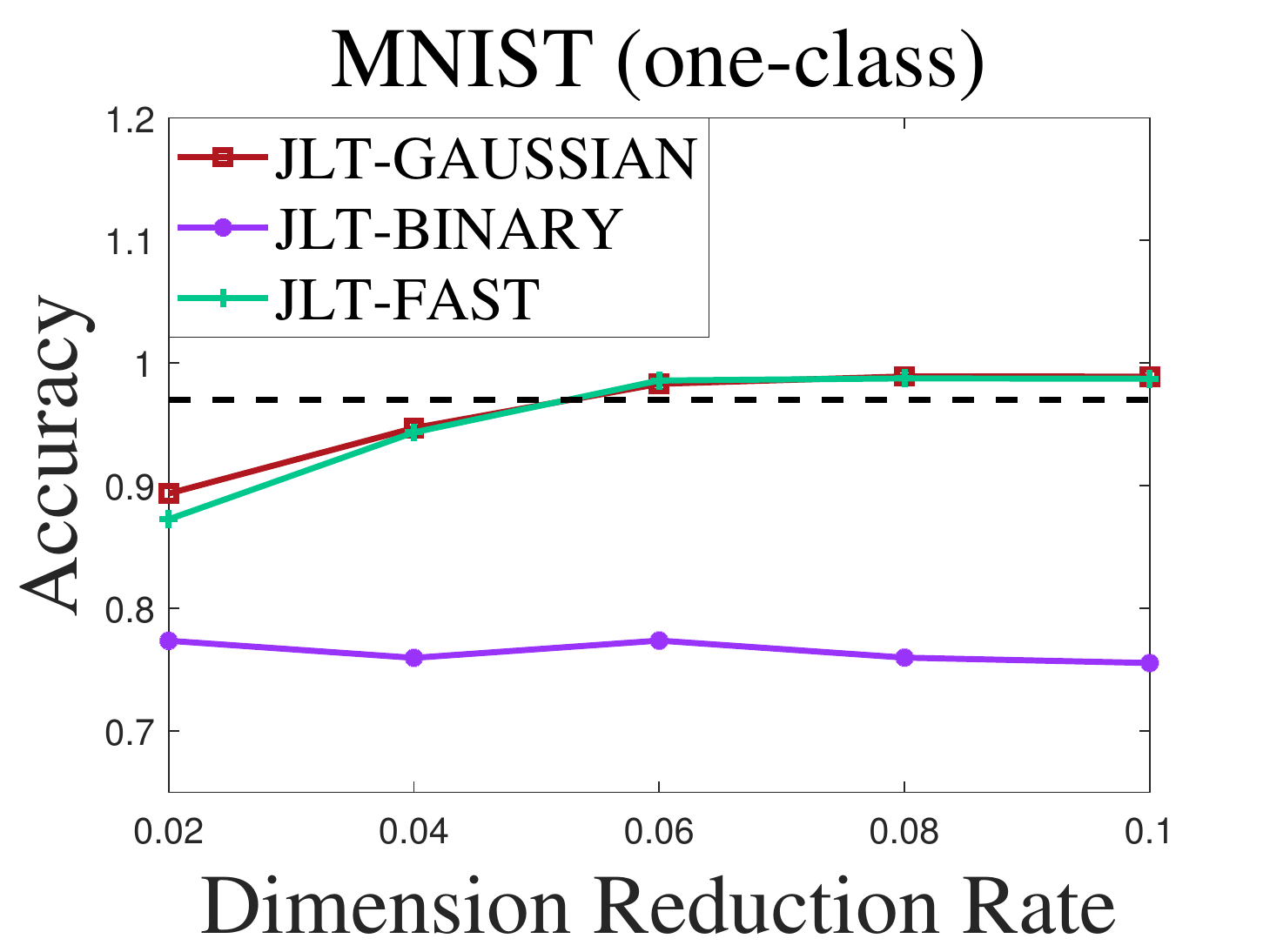}
			\includegraphics[width=0.5\columnwidth]{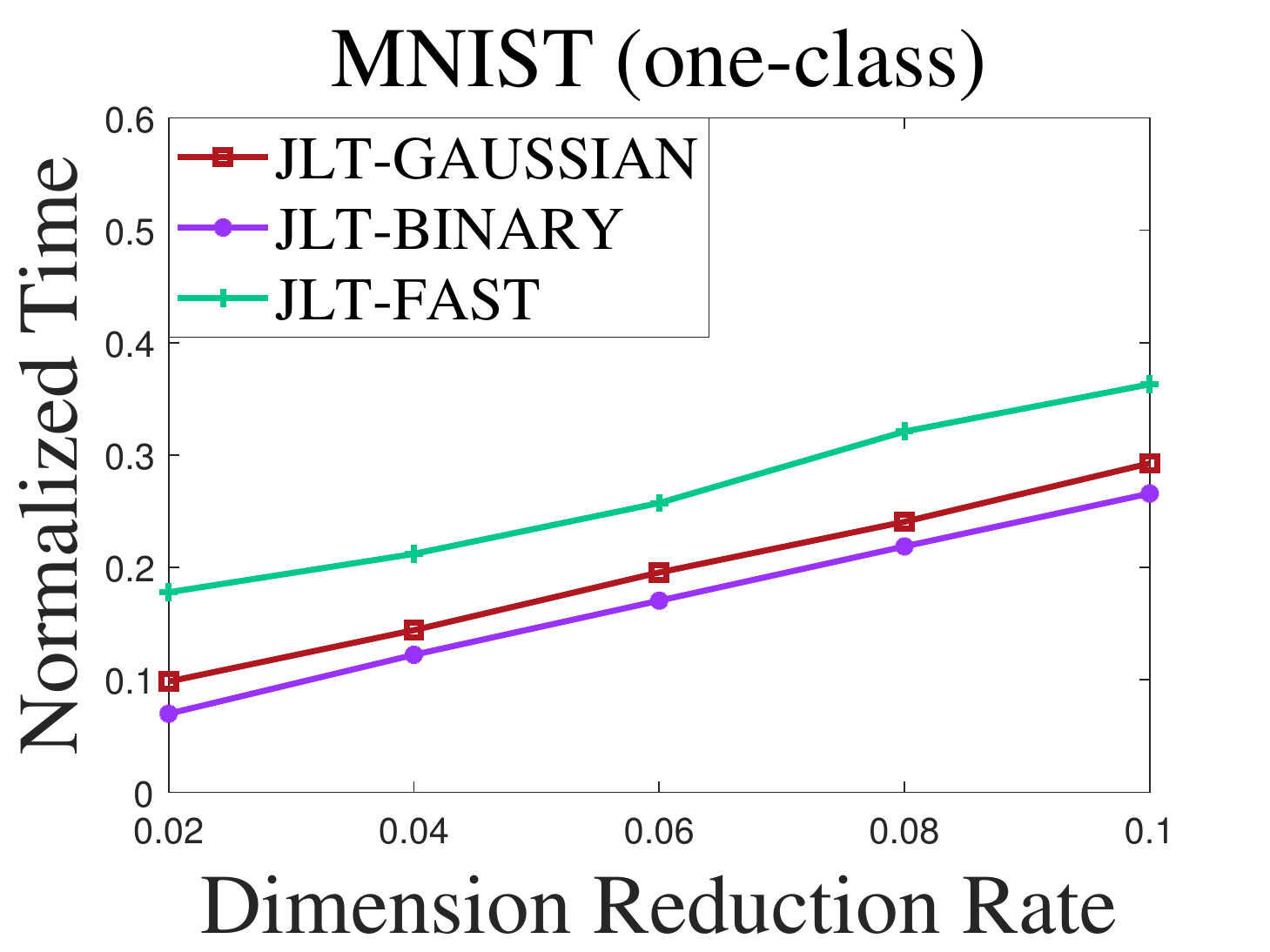}
			\includegraphics[width=0.5\columnwidth]{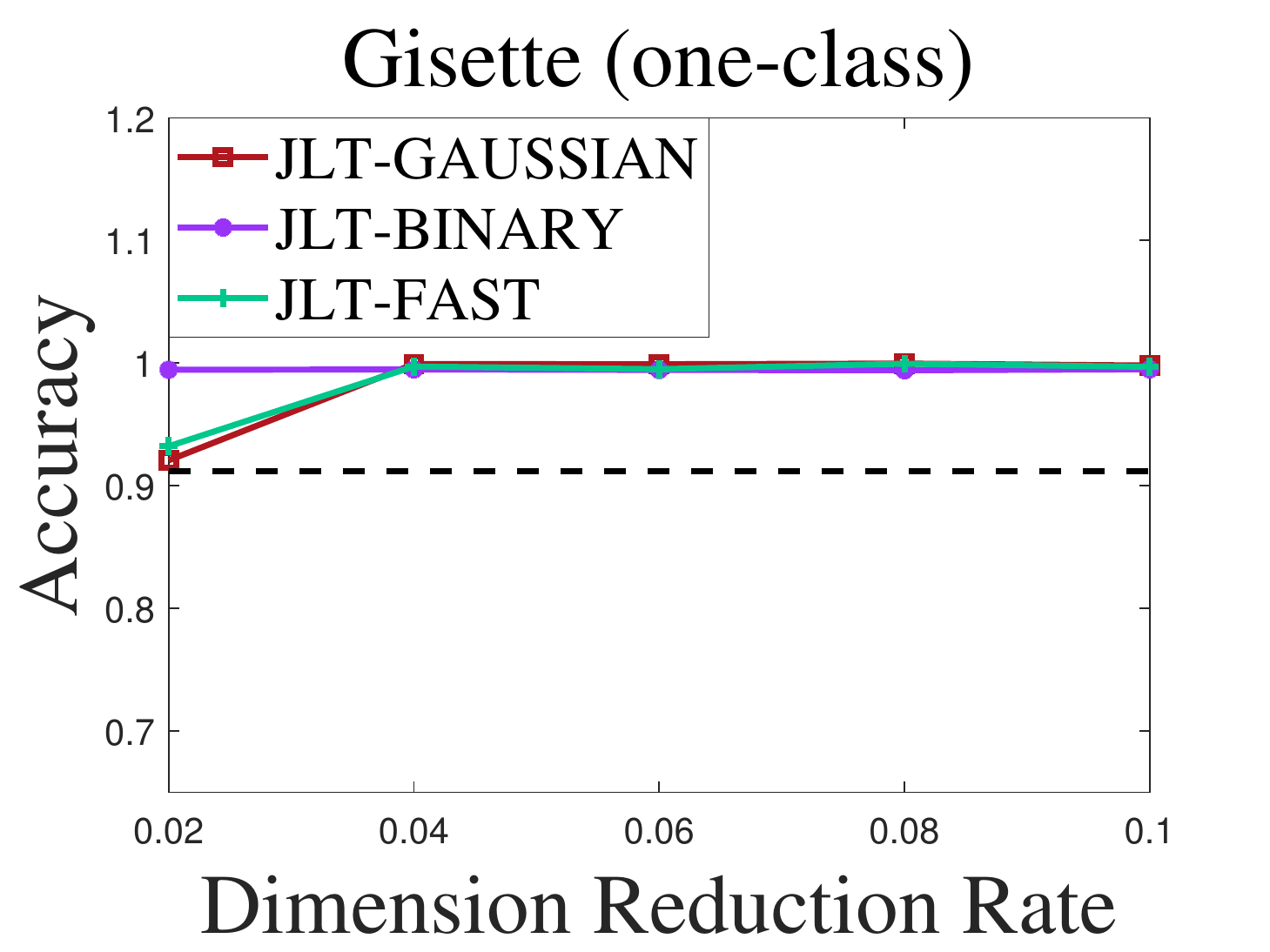}
			\includegraphics[width=0.5\columnwidth]{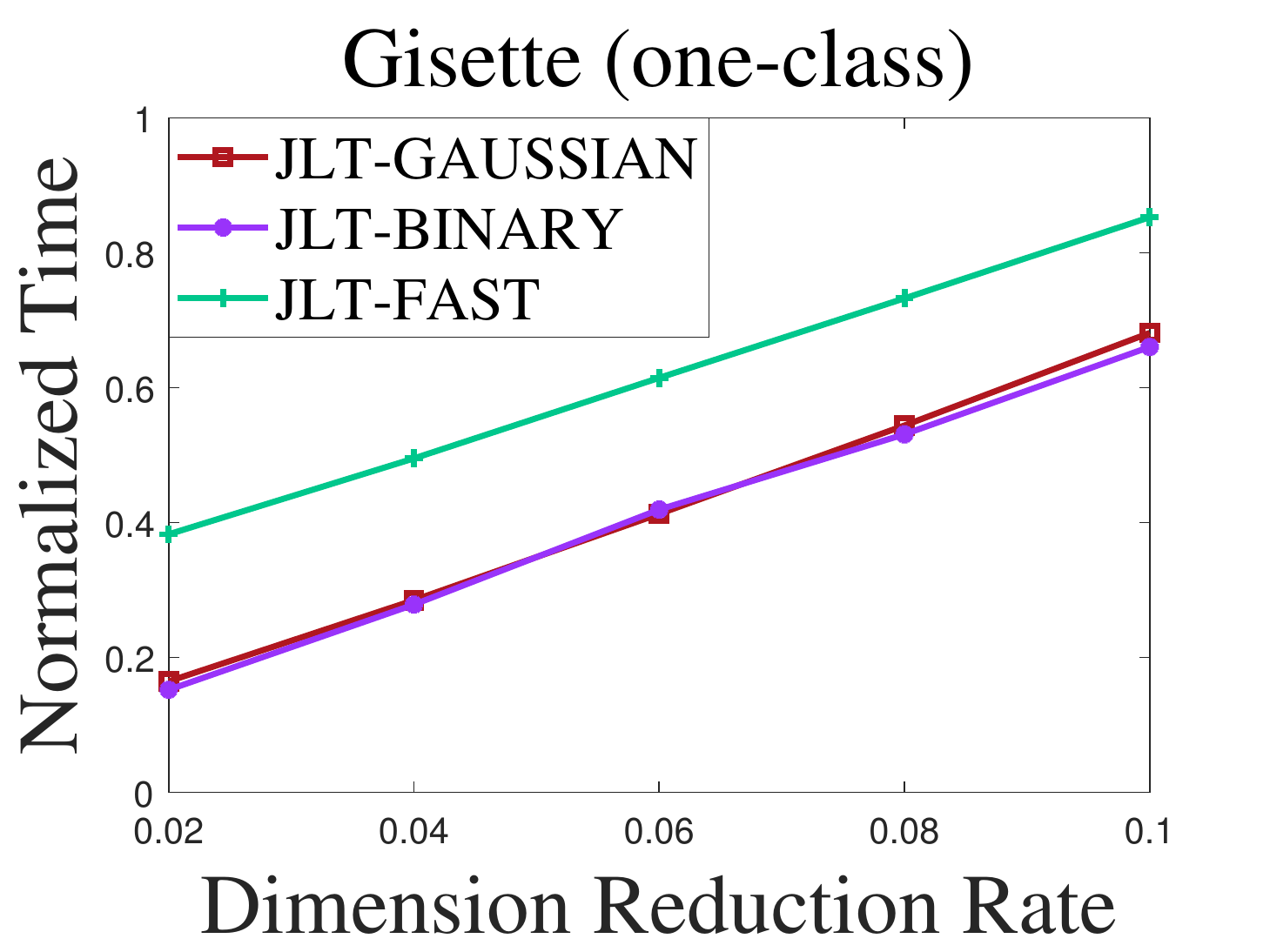}
		}
		\centerline{ \hspace{-0.8in}\hfill \stepcounter{sd2} (\alph{sd2})\hfill \stepcounter{sd2} (\alph{sd2})\hfill \stepcounter{sd2} (\alph{sd2})\hfill \stepcounter{sd2} (\alph{sd2})\hspace{0.8in}}
		%\caption{}
		%\label{icml-historical2}
	\end{center}
	\caption{The comparison on different JL transform methods for one-class SVM with outliers. The horizontal lines in (a) and (c) indicate the obtained classification accuracies without dimension reduction. }
	\label{fig-exp2}
\end{figure*}

\textbf{$k$-center clustering with outliers.} We use the popular $3$-approximation algorithm~\citep{charikar2001algorithms} as the algorithm $\mathcal{A}$. 
We consider two real-world datasets. 
\textbf{CIFAR-10} ~\citep{Cifar10} consists of $n=60,000$ color images of $k=10$ classes ({\em e.g.}, airplane, bird), with each class having $6000$ images and each image being represented by a $3072$-dimensional vector. 
\textbf{MNIST}~\citep{lecun98} contains $n=60,000$ handwritten digit images from $0$ to $9$ ({\em i.e.,} $k=10$), where each image is represented by a $784$-dimensional vector.  
For each dataset, we randomly add $10\%n$ outliers outside the MEBs of the clusters (we use the aforementioned BC's algorithm from~\citep{badoiu2003smaller} to compute an approximate MEB for each cluster). 
%We set the dimension reduction rate to be $\{6\%, 8\%, 10\%, 12\%, 14\%\}$.

We compare the three different JL transforms and show the resulting normalized radius and running time in Figure~\ref{fig-exp9}. The ``normalized radius'' ({\em resp.,} ``normalized running time'') indicates the obtained radius ({\em resp.,} running time) normalized over the one without dimension reduction. The experimental results suggest that   \textsc{JLT-Gaussian} and \textsc{JLT-fast} achieve better clustering qualities than \textsc{JLT-binary}; but \textsc{JLT-binary} is faster than the other two methods, due to its simplicity by using random $\pm 1$ entries. 

%The results suggest that \textsc{JLT-binary} is the stablest and fastest one, though it achieves slightly worse normalized radius comparing with the other two methods.

%We also compare \textsc{JLT-binary} with the other three dimension reduction methods.

\newcounter{sd9}
\begin{figure*}[ht]
	\begin{center}
		\centerline{\includegraphics[width=0.5\columnwidth]{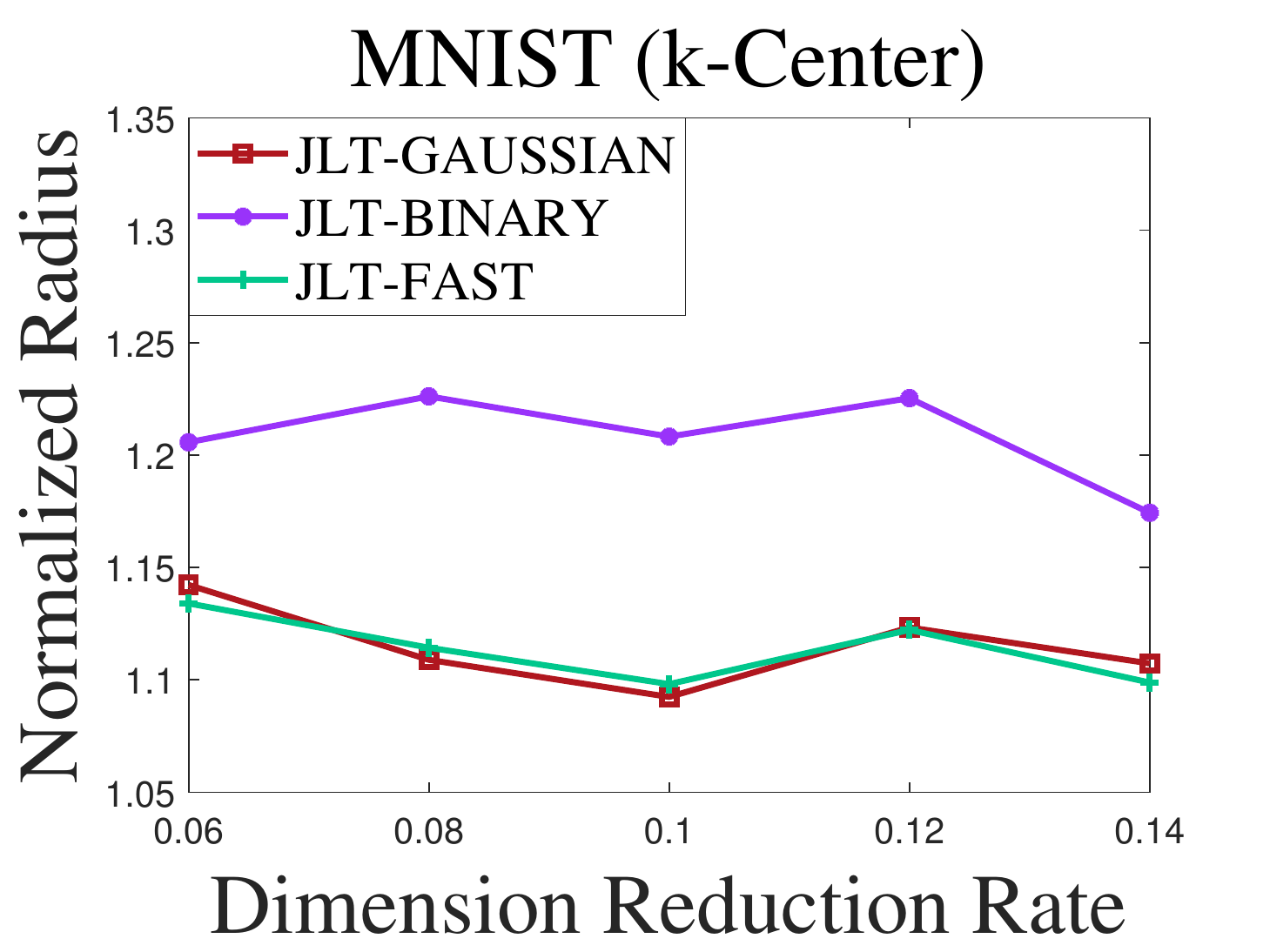}
			\includegraphics[width=0.5\columnwidth]{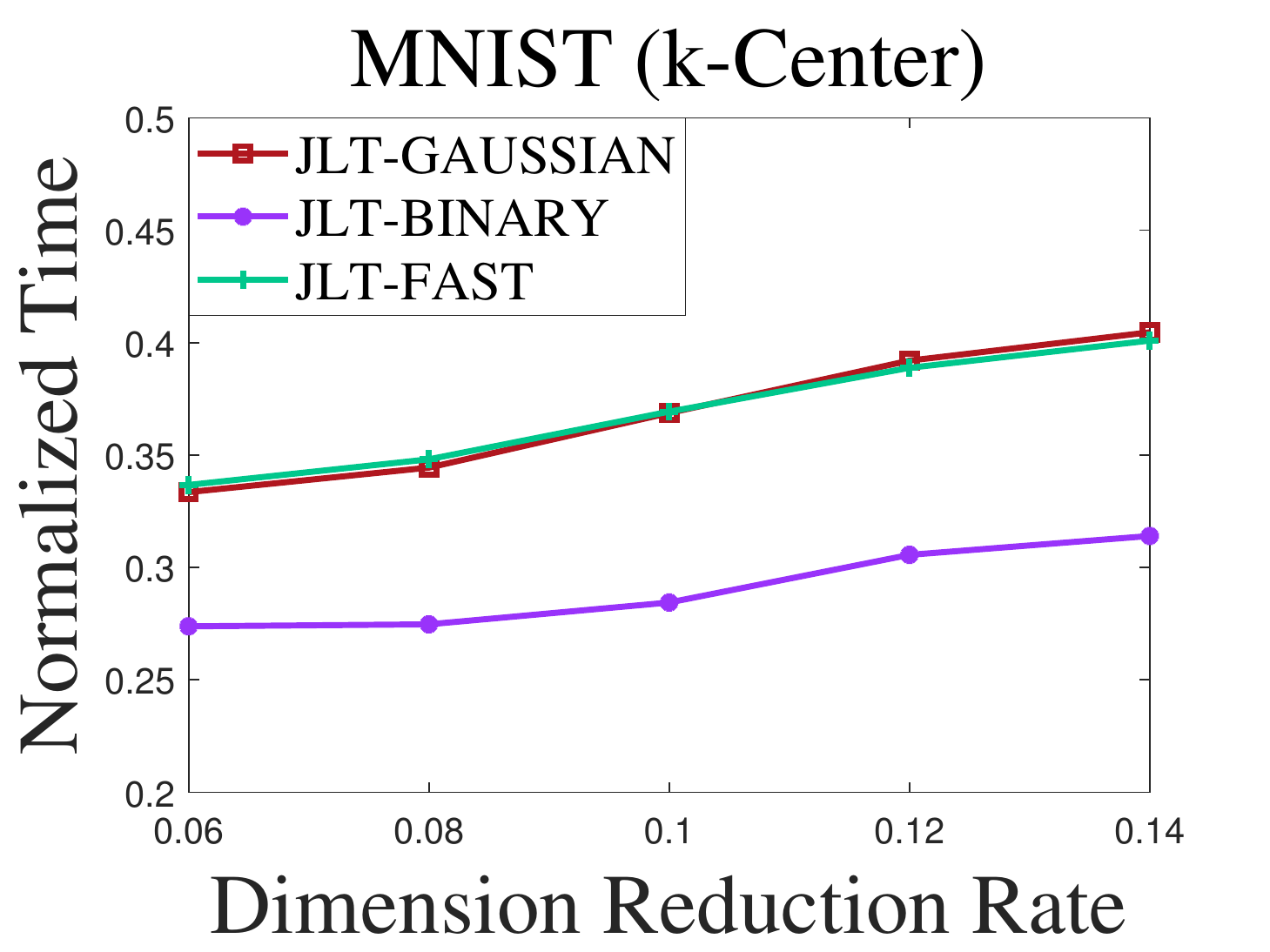}
						\includegraphics[width=0.5\columnwidth]{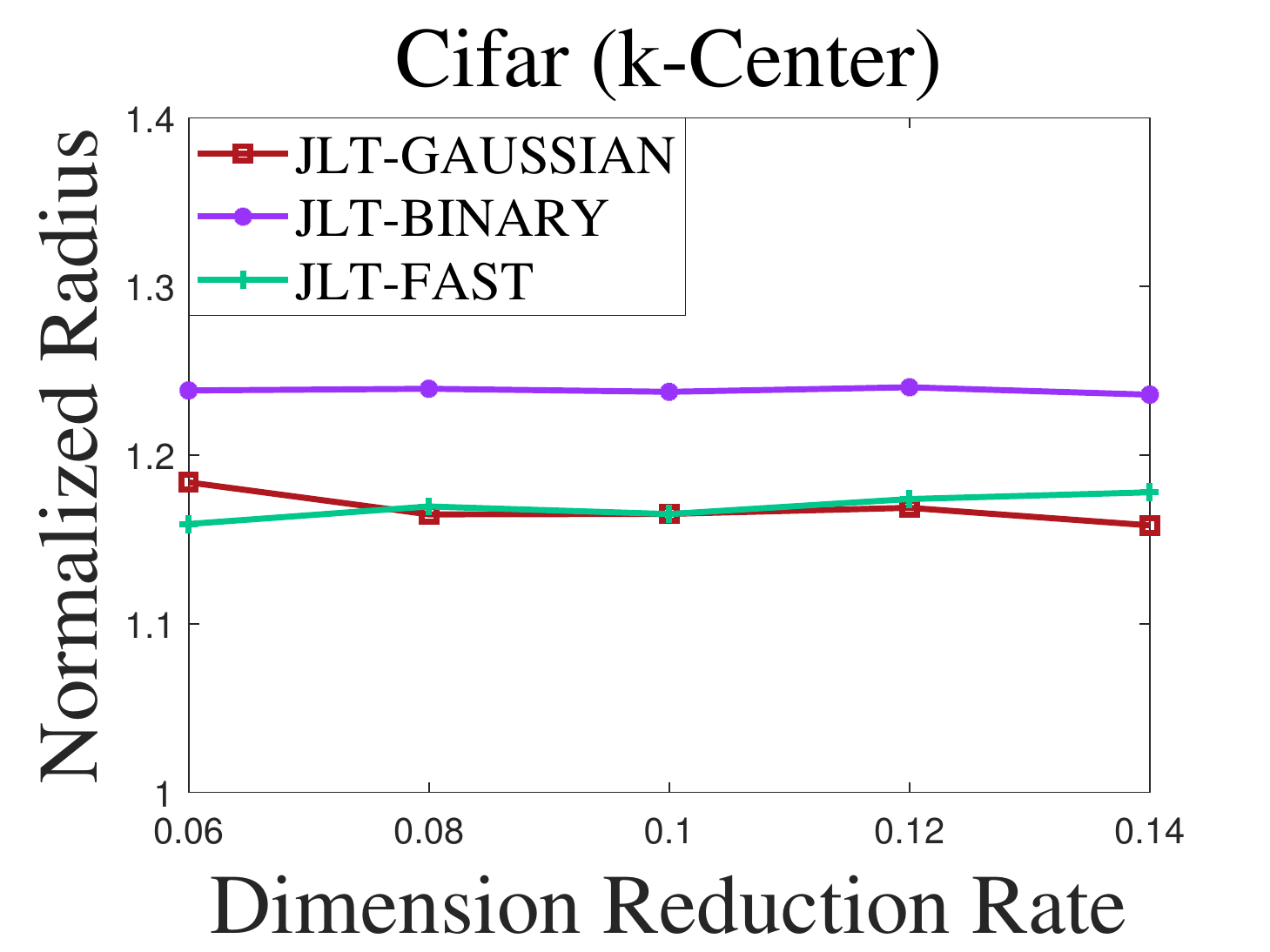}
			\includegraphics[width=0.5\columnwidth]{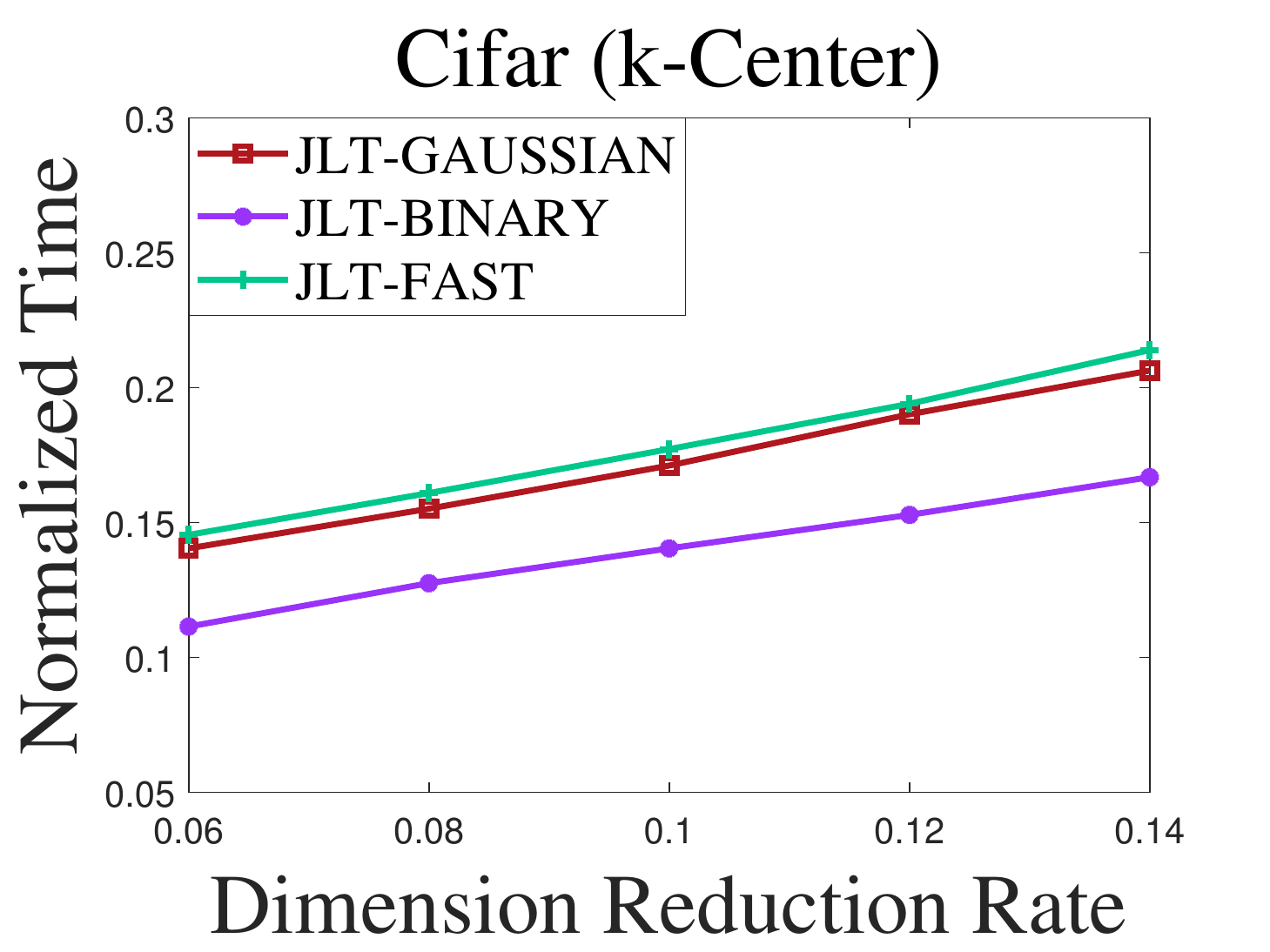}
		}
		\centerline{ \hspace{-0.8in}\hfill \stepcounter{sd9} (\alph{sd9})\hfill \stepcounter{sd9} (\alph{sd9})\hfill \stepcounter{sd9} (\alph{sd9})\hfill \stepcounter{sd9} (\alph{sd9})\hspace{0.8in}}
		%\caption{}
		%\label{icml-historical2}
	\end{center}
	\caption{The comparison on different JL transform methods for $k$-center clustering with outliers.}
	\label{fig-exp9}
\end{figure*}

\end{document}